\documentclass[11pt]{article}
\usepackage{fullpage}

\usepackage{times}
\usepackage{comment,amsfonts,amssymb,amsmath,amsthm,graphicx,algorithm,algorithmic}
\newcommand{\commentout}[1]{}

\usepackage[utf8]{inputenc}
\usepackage{titlesec}
\usepackage{amsmath}
\usepackage{enumitem}
\usepackage{amsfonts}
\usepackage{amsthm}
\usepackage{mathtools}
\usepackage{nameref}
\usepackage{graphicx}
\graphicspath{ {Images/} }
\usepackage{blindtext}

\usepackage{comment,amsfonts,amsmath,amsthm,graphicx,
algorithmic,mathtools}

\newtheorem*{theorem*}{Theorem}
\newtheorem*{lemma*}{Lemma}
\newtheorem*{corollary*}{Corollary}

\ifx\pdftexversion\undefined
\usepackage[colorlinks,linkcolor=black,filecolor=black,citecolor=black,urlco
lor=black,pdfstartview=FitH]{hyperref}
\else
\usepackage[colorlinks,linkcolor=blue,filecolor=blue,citecolor=blue,urlcolor
=blue,pdfstartview=FitH,linktocpage=true]{hyperref}
\fi

\def\MathF{\hbox{\rm I\kern-2pt F}}
\def\MathP{\hbox{\rm I\kern-2pt P}}
\def\MathR{\hbox{\rm I\kern-2pt R}}
\def\MathZ{\hbox{\sf Z\kern-4pt Z}}
\def\MathN{\hbox{\rm I\kern-2pt I\kern-3.1pt N}}
\def\MathC{\hbox{\rm \kern0.7pt\raise0.8pt\hbox{\footnotesize I}
		\kern-4.2pt C}}
\def\MathQ{\hbox{\rm I\kern-6pt Q}}

\newcommand{\mommit}[1]{}
\newcommand{\namedref}[2]{\hyperref[#2]{#1~\ref*{#2}}}

\newtheorem{theorem}{Theorem}
\newtheorem{lemma}{Lemma}

\newtheorem{definition}{Definition}

\usepackage{pdfsync}
\usepackage{authblk}

\usepackage{authblk}

\begin{document}
\title{Path-Reporting Distance Oracles with Linear Size}
\author{Ofer Neiman\footnote{Ben-Gurion University of the Negev, Israel. E-mail: \texttt{neimano@cs.bgu.ac.il}}\qquad  Idan Shabat\footnote{Ben-Gurion University of the Negev, Israel. Email:\texttt{shabati@post.bgu.ac.il}}}
\date{}

\maketitle

\begin{abstract}

Given an undirected weighted graph, an (approximate) distance oracle is a data structure that can (approximately) answer distance queries. A {\em Path-Reporting Distance Oracle}, or {\em PRDO}, is a distance oracle that must also return a path between the queried vertices.
Given a graph on $n$ vertices and an integer parameter $k\ge 1$, Thorup and Zwick \cite{TZ01} showed a PRDO with stretch $2k-1$, size $O(k\cdot n^{1+1/k})$ and query time $O(k)$ (for the query time of PRDOs, we omit the time needed to report the path itself). Subsequent works \cite{MN06,C14,C15} improved the size to $O(n^{1+1/k})$ and the query time to $O(1)$. However, these improvements produce distance oracles which are not path-reporting. Several other works \cite{ENW16,EP15} focused on small size PRDO for general graphs, but all known results on distance oracles with linear size suffer from polynomial stretch, polynomial query time, or not being path-reporting.

In this paper we devise the first linear size PRDO with poly-logarithmic stretch and low query time $O(\log\log n)$. More generally, for any integer $k\ge 1$, we obtain a PRDO with stretch at most $O(k^{4.82})$, size $O(n^{1+1/k})$, and query time $O(\log k)$. In addition, we can make the size of our PRDO as small as $n+o(n)$, at the cost of increasing the query time to poly-logarithmic. For unweighted graphs, we  improve the stretch to $O(k^2)$.

We also consider {\em pairwise PRDO}, which is a PRDO that is only required to answer queries from a given set of pairs ${\cal P}$. An exact PRDO of size $O(n+|{\cal P}|^2)$ and constant query time was provided in \cite{EP15}. In this work we dramatically improve the size, at the cost of slightly increasing the stretch. Specifically, given any $\epsilon>0$, we devise a pairwise PRDO with stretch $1+\epsilon$, constant query time, and near optimal size $n^{o(1)}\cdot (n+|{\cal P}|)$.

\end{abstract}

\newpage
\tableofcontents
\newpage

\section{Introduction}

Given an undirected weighted graph $G=(V,E)$ with $n$ vertices and positive weights on the edges $w:E\to\mathbb{R}_+$, the distance $d_G(u,v)$ between two vertices $u,v\in V$ is the minimal weight of a path between them in $G$. For a parameter $\alpha\ge 1$, a {\em distance oracle} with stretch $\alpha$ is a data structure, that given a query for a pair of vertices $(u,v)$, returns an estimated distance $\hat{d}(u,v)$ such that
\[
d_G(u,v)\leq\hat{d}(u,v)\leq\alpha\cdot d_G(u,v)~.
\]
A {\em Path-Reporting Distance Oracle}, or {\em PRDO}, is a distance oracle that must also return a path in $G$ of weight $\hat{d}(u,v)$ between the queried vertices $u,v$.

Distance oracles have been the subject of extensive research in the last few decades. They are fundamental objects in Graph Algorithms, due to their both practical and theoretical usefulness. The main interest is in the triple tradeoff between the stretch of a distance oracle, its size, and its query time. In some cases, the preprocessing time, (that is, the time needed to construct the distance oracle) is also considered. Note that for every query, a PRDO must return a path $P$, thus the running time of the query algorithm is always in the general form of $O(q+|P|)$. We usually omit the term\footnote{Throughout this paper, $|P|$ denotes the number of \textit{edges} in a path $P$.} $|P|$ from the query time and write only $O(q)$.

This work focuses on path-reporting distance oracles for general graphs. The path-reporting property is more appealing for certain applications that require navigation or routing \cite{JHR96,DSSW09,Z08}. See the survey \cite{S14} and the references therein for additional applications of distance oracles.

\subsection{Linear Size Path-Reporting Distance Oracles}

In \cite{TZ01}, a PRDO with stretch $2k-1$, size $O(kn^{1+1/k})$ and query time $O(k)$ was shown. Assuming the girth conjecture of Erd\H{o}s \cite{E64}, this result is best possible, up to the factor of $k$ in the size and query time.\footnote{In fact, Erd\H{o}s girth conjecture only implies that to achieve stretch $2k-1$, any distance oracle must use $\Omega(n^{1+\frac{1}{k}})$ \textit{bits}. However, in \cite{C15} a lower bound of $\Omega(n^{1+\frac{1}{k}})$ \textit{words} (each of size $\log n$ bits) for PRDOs is proved. 
} The query time was improved to $O(\log k)$ by \cite{WN13}. Observe that $kn^{1/k}\ge\log n$ for any $k\ge 1$, so the PRDO of \cite{TZ01,WN13} cannot be sparser than $\Theta(n\log n)$.
Subsequent works \cite{MN06,C14,C15} obtained distance oracles with stretch $2k-1$, improved size $O(n^{1+1/k})$ and constant query time. However, these distance oracles are not path-reporting.

Of particular interest is trying to achieve a PRDO of linear size. The first such result \cite{ENW16} obtained a PRDO with size $O(tn)$, for any parameter $t\ge 1$, and $O(\log t)$ query time, but had a polynomial stretch $O(tn^{2/\sqrt{t}})$, and required that the aspect ratio of the weights is polynomially bounded. This result was improved by \cite{EP15}, who showed a PRDO with stretch $O(k)$ and size $O(n^{1+1/k})$, but at the cost of increasing the query time to $O(n^{1/k+\epsilon})$, where $\epsilon>0$ is a {\em constant}. Note that the query time $O(n^{1/k+\epsilon})$ is prohibitively large - this term dominates the length of many of the output paths, so the PRDO suffers from large query time for these paths. For this reason, it is of special interest to construct linear size PRDOs with query time that is far less than polynomial in $n$, say polylogarithmic in $n$. Another variant of \cite{EP15} does achieve a PRDO with query time $O(\log\log n)$ and stretch $\log^{O(1)}n$, however its size $O(n\log\log n)$ is no longer linear.

We conclude that in all previous results, every linear size distance oracle suffers from a polynomial stretch \cite{ENW16}, has polynomial query time \cite{EP15}, or simply cannot report paths \cite{C15}.

\subsubsection{Our Results}

In this work we devise the first linear size PRDO for general graphs with polylogarithmic stretch and low query time. Specifically, for any integer $k\ge 1$ our PRDO has stretch $O(k^{4.82})$, size $O(n^{1+1/k})$, and query time $O(\log k)$. Indeed, setting $k=\log n$ yields linear size, $\log^{O(1)}n$ stretch and $O(\log\log n)$ query time. Our main result is for the case $k=\log n$, where we get a linear size PRDO with low query time. In fact, for any $k>\frac{\log n}{\log\log\log n}$, our new PRDO improves all previous bounds.

Note that since the query time is low, in most cases it is dominated by the length of the reported path. Therefore, the strength of this result is not in the precise expression for the query time, but in the fact that the query time is far less than polynomial in $n$.

We can refine our result to obtain an {\em ultra-compact} PRDO, whose size is as small as $n+o(n)$ (we measure the size by \textit{words}, that is, the oracle uses storage of $n\log n+o(n\log n)$ bits), at the cost of increasing the query time to $\log^{O(1)}n$. In view of the lower bound of \cite{C15}, this space usage is optimal, up to additive lower order terms.
If the graph is unweighted, we offer a simpler construction with improved stretch $O(k^2)$.

\subsection{Pairwise Path-Reporting Distance Oracles}

A {\em pairwise distance oracle} is a distance oracle that is also given as input a set of pairs ${\cal P}$, and is required to answer queries only for pairs in ${\cal P}$. The problem of designing such oracles is related to the extensive research on {\em distance preservers}: these are subgraphs that preserve exactly all distances between pairs in ${\cal P}$. When allowing some stretch, these are sometimes called {\em pairwise spanners}.\footnote{In \cite{KP22} pairwise spanners with small $1+\epsilon$ stretch are called {\em near-exact preservers}.} Distance preservers were introduced in \cite{CE05}, and pairwise spanners have been studied in \cite{CGK13,K17,BCPS20,B21,BW21,KP22}.

In \cite{EP15}, an exact pairwise PRDO was shown with constant query time and size $O(n+|{\cal P}|^2)$. For distance preservers, \cite{CE05} showed a lower bound of $\Omega(n^{2/3}|{\cal P}|^{2/3})$. Note that for $|{\cal P}|=n^{2-\delta}$, the lower bound implies that any distance preserver must have size $\Omega(|{\cal P}|\cdot n^{\delta/3})$, so there are no distance preservers with near-linear size (except for the trivial case when ${\cal P}$ contains a constant fraction of all pairs).
In \cite{BHT22}, it was proved that exact pairwise PRDOs suffer from the same lower bounds of exact preservers (see Theorem 14 in \cite{BHT22}). Thus, it is very natural to ask if the size of a pairwise PRDO can be reduced when allowing a small stretch. Specifically, we would like to obtain a very small stretch (e.g., $1+\epsilon$ for any $\epsilon>0$), and size that is proportional to $|{\cal P}|+n$ (which is the basic lower bound).

\subsubsection{Our Results}
In this work we devise a pairwise PRDO with near optimal size $n^{o(1)}\cdot (|{\cal P}|+n)$, constant query time, and stretch $1+\epsilon$, for any $\epsilon>0$ (the $o(1)$ term in the size depends logarithmically on $\epsilon$). This result uses the techniques of \cite{KP22} on hopsets and spanners, and extends them for pairwise path-reporting distance oracles.

\subsection{Our Techniques}

Our main result on linear size (and ultra-compact) PRDO uses a conceptually simple idea: we partition the graph into $O(n/k)$ clusters, and define the {\em  cluster-graph} by contracting every cluster to a single vertex (keeping the lightest edge among parallel edges). Next, we apply the \cite{TZ01} distance oracle on this cluster-graph. In addition, we store a certain spanning tree for every cluster. Given a query $(u,v)$, the algorithm first finds a path in the cluster-graph between the clusters containing $u,v$, and for each cluster in this path, it finds an inner-cluster path between the entry vertex and the exit vertex of the cluster, using the spanning tree.

In order to implement this framework, it is required to find a clustering so that the overhead created by going through the spanning tree of every cluster is small enough. For unweighted graphs, we apply a simple clustering with radius $O(k)$, and maintain a BFS tree for each cluster. However, for weighted graphs a more intricate clustering is required. Note that we cannot enforce a small diameter bound on all clusters, since by only restricting the diameter, each cluster can be very small, and we need to have at most $O(n/k)$ clusters. Instead, we use a variant of Bor\r{u}vka clustering \cite{BBDHKLM17}.

In Bor\r{u}vka's algorithm for minimum spanning tree, in each phase, every vertex adds its adjacent edge of minimal weight to a forest $F$ (breaking ties consistently), and the connected components of $F$ are contracted to yield the vertices of the next phase. If we truncate this process after $t$ phases, we get a clustering, analyzed in \cite{BBDHKLM17}. Unfortunately, any phase in this clustering may produce long chains, in which case the stretch cannot be controlled. To rectify this, we delete certain edges in $F$, so that every cluster is a star, while ensuring that the number of non-singleton clusters is large enough. These stars are also the basis for the spanning tree of each cluster. The main technical part is analyzing the stretch induced by this clustering on the paths returned by calling the distance oracle on the cluster-graph.

\subsection{Organization}

After some preliminaries in Section~\ref{sec:Preliminaries}, we show our PRDO for  unweighted graph is in Section~\ref{sec:Unweighted}, and for weighted graphs in Section~\ref{sec:NewPRDOSec}. Our result for pairwise PRDO appears in Section~\ref{sec:PairwisePRDO}. 

\subsection{Bibliographic Note}
Following this work, \cite{ES23} showed (among other results) a PRDO of linear size with stretch $\tilde{O}(\log n)$ and query time $O(\log\log\log n)$.

\section{Preliminaries} \label{sec:Preliminaries}

Let $G=(V,E)$ be an undirected weighted graph. In all that follows we assume that $G$ is connected.

\noindent{\bf Spanners.}
For a parameter $\alpha\ge 1$, an {\em $\alpha$-spanner} is a subgraph $S$ of $G$, such that for every two vertices $u,v\in V$,
\begin{equation}\label{eq:span}
d_S(u,v)\leq\alpha\cdot d_G(u,v)~.
\end{equation}
The spanner is called a {\em pairwise} spanner, if for a given a set of pairs ${\cal P}$, we only require \eqref{eq:span} to hold for all $(u,v)\in {\cal P}$.

\noindent{\bf Trees.}
Let $x,y$ be two vertices in a rooted tree. We denote by $p(x)$ the {\em parent} of $x$, which is the unique neighbor of $x$ that lies on the path from $x$ to the root, and by $h(x)$ its {\em height}, which is the number of edges on the path from $x$ to the root. Denote by $lca(x,y)$ the lowest common ancestor of $x,y$, which is a vertex $z$ such that $x,y$ are both in its sub-tree, but not both in the sub-tree of any child of $z$. Note that the unique path between $x,y$ in the tree is the concatenation of the unique paths from $x$ to $lca(x,y)$, and from $lca(x,y)$ to $y$.

The following lemma let us easily find paths within a tree $T$.
\begin{lemma} \label{lemma:TreeRouting}
Let $T$ be a rooted tree, and assume we are given $h(x)$ and $p(x)$ for every vertex $x$ in $T$. There is an algorithm that given two distinct vertices $a,b$ in a tree $T$, finds the unique path between $a,b$ in $T$. The running time of this algorithm is proportional to the number of edges in the output path (or $O(1)$ if the path is empty).
\end{lemma}

\begin{proof}

First, if $a=b$, then the desired path is empty and we return it in $O(1)$ time.
Otherwise, if $h(a)>h(b)$, recursively find the unique path $P_{p(a),b}$ in $T$ between $p(a)$ and $b$, and return $\{a,p(a)\}\circ P_{p(a),b}$. Symmetrically, if $h(b)\ge h(a)$, return $\{b,p(b)\}\circ P_{p(b),a}$.

The correctness of this algorithm follows from the fact that if, for example, $h(a)\ge h(b)$, and $a\neq b$, then $b$ cannot be in the sub-tree of $a$ in $T$, hence the unique path between $a,b$ must pass through $p(a)$.
In each recursive call we reduce the sum $h(a)+h(b)$ by $1$, and therefore the algorithm ends when $a=b=lca(a,b)$. Therefore the running time of this algorithm is proportional to the length of the unique paths from $a$ to $lca(a,b)$ and from $b$ to $lca(a,b)$. The concatenation of these paths is exactly the returned path by our algorithm, which is the unique path in $T$ between $a,b$. Hence, the running time is proportional to the number of edges in the output path, as desired.

\end{proof}

\subsection{Thorup-Zwick PRDO}

A main component of our new PRDO relies on a well-known construction by Thorup and Zwick \cite{TZ01}. Given a weighted graph $G$ with $n$ vertices and an integer parameter $k\geq1$, they constructed a PRDO with stretch $2k-1$, query time $O(k)$ and size $O(kn^{1+1/k})$.

A useful property of the Thorup-Zwick (TZ) PRDO is that for every query, it returns a path that is contained in a sub-graph $S$ of $G$, such that $|S|=O(kn^{1+1/k})$. Notice that since the stretch of this PRDO is $2k-1$, then $S$ must be a $(2k-1)$-spanner of $G$. We call $S$ the \textit{underlying spanner} of the PRDO. One can compute the underlying spanner $S$ either during the PRDO construction, or after its construction by querying the PRDO on every pair of vertices, and computing the union of the resulting paths.

A result of \cite{WN13} improved the query time of the TZ PRDO to $O(\log k)$ instead of $O(k)$, while returning the same path that the TZ PRDO returns. Indeed, when we use here the TZ PRDO, we consider its query time to be $O(\log k)$.
The following theorem concludes this discussion.
\begin{theorem}[By \cite{TZ01} and \cite{WN13}] \label{thm:TZ}
Let $G$ be an undirected weighted graph with $n$ vertices, and let $k\geq1$ be an integer parameter. There is a PRDO for $G$ with stretch $2k-1$, query time $O(\log k)$ and size $O(kn^{1+1/k})$, with an underlying spanner of the same size.
\end{theorem}

\section{Path-Reporting Distance Oracle for Unweighted Graphs} \label{sec:Unweighted}

In this section we introduce a simple variant of our construction, tailored for unweighted graphs. We first apply a simple clustering, and store a BFS (Breadth First Search) tree for each cluster. We next apply the TZ PRDO on the resulting cluster-graph. Finally, each query $(u,v)$ is answered by taking the path in the cluster-graph between the clusters containing $u$ and $v$, and completing it to a path in $G$ using the BFS trees.


\subsection{Clustering}

We start by dividing the graph into clusters, using the following lemma.
\begin{lemma} \label{lemma:TreePartition}
Let $G=(V,E)$ be an undirected unweighted graph with $n$ vertices. Let $k\in[1,n]$ be some integer. There is an algorithm that finds a partition $V=\bigcup_{i=1}^qC_i$, such that every $C_i$ has a spanning tree $T_i=(C_i,E_i)$ with root $r_i$, where $E_i\subseteq E$ and for every $v\in C_i$, $d_{T_i}(v,r_i)\leq k$. In addition, the number of sets in this partition, $q$, is at most $\frac{n}{k}$.
\end{lemma}

\begin{proof}

Fix some $r\in V$, and let $T=(V,E_T)$ be the BFS tree with $r$ as a root. The tree $T$ is actually the shortest paths tree from $r$ in $G$, and so the path from every $v\in V$ to $r$ in $T$ is of length exactly $d_G(v,r)$, i.e., $d_T(v,r)=d_G(v,r)$. If every vertex $v\in V$ satisfies $d_G(v,r)\leq k$, then we can return the trivial partition $\{V\}$, with the spanning tree $T$ and root $r$.

Otherwise, let $v$ be the furthest leaf of $T$ from $r$, that is, $v$ maximizes the length $d_T(v,r)$. We know that $d_T(v,r)>k$, and since $G$ is unweighted, there is a vertex $r'$ on the path in $T$ from $v$ to $r$, with $d_T(v,r')=k$. Denote by $T'$ the sub-tree of $T$ rooted at $r'$.

Let $u\in V$ be a vertex in $T'$. Since $r'$ is on the path from $u$ to $r$, and on the path from $v$ to $r$, we have
\[
d_{T'}(u,r')=d_T(u,r)-d_T(r',r)\leq d_T(v,r)-d_T(r',r)=d_T(v,r')=k~.
\]
Therefore, if $C$ is the set of vertices of $T'$, we can return $C$ as one of the sets in the desired partition, where its spanning tree is $T'$ and its root is $r'$. We then delete $C$ from $G$ and continue recursively.

Note that the tree $T'$ contains the path from $v$ to $r'$, which is of length $k$. Since $G$ is unweighted, that means that $T'$ contains at least $k$ vertices, and so does $C$. Hence, the number of vertices in the graph, after the deletion of $C$, is at most $n-k$. Notice also that the tree $T$ is still a tree, after the removal of $T'$, thus the remaining graph is still connected. As a result, we can assume that our algorithm recursively partitions the remaining graph into at most $\frac{n-k}{k}=\frac{n}{k}-1$ sets, with spanning trees and roots as desired. Together with the last set $C$, we obtain a partition into at most $\frac{n}{k}$ parts, with the wanted properties.

\end{proof}

Given the unweighted graph $G=(V,E)$ and the integer $k$, let $\mathcal{C}$ be a partition as in Lemma \ref{lemma:TreePartition}. For every $C\in\mathcal{C}$, let $T[C]$ and $r[C]$ be the spanning tree of $C$ and its root. We define a new graph $\mathcal{H}=(\mathcal{C},\mathcal{E})$ as follows.

\begin{definition} \label{def:ClustersGraph}
The graph $\mathcal{H}=(\mathcal{C},\mathcal{E})$ is defined as follows. The set $\mathcal{E}$ consists of all the pairs $\{C,C'\}$, where $C,C'\in\mathcal{C}$, such that there is an edge in $G$ between $C,C'$.

Given an edge $\{C,C'\}\in\mathcal{E}$, we denote by $e(C,C')$ the edge $\{x,y\}$ of $G$, where $\{x,y\}$ is some choice of an edge that satisfies $x\in C$, $y\in C'$.
\end{definition}

We denote by $F$ the forest that consists of the disjoint union of the trees $T[C]$, for every $C\in\mathcal{C}$. For a vertex $x\in V$, define $h(x)$ to be the height of $x$ in the tree $T[C]$ such that $x\in C$, and $p(x)$ its parent in this tree.

\subsection{Stretch Analysis}

Fix any cluster $C$, let $T=T[C]$ be its spanning tree with root $r=r[C]$. For any two vertices $a,b\in T$, the unique path between them is a sub-path of the union between the two paths from $a$ to $r$ and from $b$ to $r$. Both of these paths are of length at most $k$. Hence, the resulting path is of length at most $2k$, and this path is exactly the one that the algorithm from Lemma \ref{lemma:TreeRouting} returns.

\begin{lemma} \label{lemma:ExtractingAlgorithm2}
There is an algorithm that given two vertices $u,v\in V$, and a simple path $Q=(C_1,C_2,...,C_t)$ in the graph $\mathcal{H}$, such that $u$ is in $C_1$ and $v$ is in $C_t$, returns a path $P$ in $G$ between $u$ and $v$, with number of edges
\[
|P|\leq t\cdot(2k+1)~.
\]

The running time of the algorithm is proportional to the number of edges in the output path. The required information for the algorithm is the set $\{h(x),p(x)\}_{x\in V}$, and the set $\{e(C_j,C_{j+1})\}_{j=1}^{t-1}$.
\end{lemma}

\begin{proof}

Given the edges $\{C_j,C_{j+1}\}$, the set $\{e(C_j,C_{j+1})\}_{j=1}^{t-1}$ can be used to find $x_j,y_j\in V$ (vertices of the original graph $G=(V,E)$), such that $x_j\in C_j$, $y_j\in C_{j+1}$ and $\{x_j,y_j\}\in E$. Define also $y_0=u, x_t=v$. For every $j\in[1,t]$, using the set $\{h(x),p(x)\}_{x\in V}$, we can use Lemma \ref{lemma:TreeRouting} to find a path $P_j$ in $G$ between $y_{j-1}$ and $x_j$, with length at most $2k$. Finding all of these paths takes time that is proportional to the sum of lengths of these paths.

The returned path by this algorithm is
\[
P=P_1\circ\{x_1,y_1\}\circ P_2\circ\{x_2,y_2\}\circ\cdots\circ\{x_{t-1},y_{t-1}\}\circ P_t~.
\]
The time needed to report this path is $O(\sum_{j=1}^t|P_j|)=O(|P|)$. The length of this path is
\[
t-1+\sum_{j=1}^t|P_j|\leq t-1+t\cdot2k<t\cdot(2k+1)~.
\]
This concludes the proof of the lemma.

\end{proof}

\subsection{A PRDO for Unweighted Graphs}
We are now ready to introduce the construction of our small size PRDO.

\begin{theorem} \label{thm:NewUnweightedPRDO}
Let $G=(V,E)$ be an undirected unweighted graph with $n$ vertices, and let $k\in[1,\log n]$ be some integer parameter. There is a path-reporting distance oracle for $G$ with stretch $2k(2k+1)=O(k^2)$, query time $O(\log k)$ and size $O(n^{1+\frac{1}{k}})$.
\end{theorem}

\begin{proof}

Denote by $TZ$ the PRDO from Theorem \ref{thm:TZ} with the parameter $k$, when constructed over the graph $\mathcal{H}=(\mathcal{C},\mathcal{E})$ (the clustering $\mathcal{C}$ is constructed with $k$ as the radius\footnote{Actually, by constructing the clustering $\mathcal{C}$ with radius $\frac{k}{8}$ instead of $k$, the stretch of our new PRDO decreases from $~4k^2$ to $k^2$. In the same way, one can achieve an arbitrarily small leading constant in the stretch.}). Let $S_{TZ}\subseteq\mathcal{E}$ be the set of the edges of the underlying spanner of $TZ$. In addition, for a given vertex $x\in V$, denote by $C(x)$ the vertex of $\mathcal{H}$ (i.e., cluster) that contains $x$. Recall also that $h(x)$ is the height of $x$ in the tree spanning $C(x)$ and $p(x)$ denotes the parent of $x$ in this tree.

We define our new PRDO for the undirected unweighted graph $G=(V,E)$. This PRDO contains the following information.
\begin{enumerate}
    \item The TZ PRDO.
    \item The set $\{e(C,C')\;|\;\{C,C'\}\in S_{TZ}\}$.
    \item The variables $\{h(x),p(x)\}_{x\in V}$.
    \item The variables $\{C(x)\}_{x\in V}$.
\end{enumerate}

Given a query $(u,v)\in V^2$, our PRDO queries $TZ$ on the vertices $C(u),C(v)$ of $\mathcal{H}$. Let $Q=(C(u)=C_1,C_2,...,C_t=C(v))$ be the resulting path, and note that all of its edges are in $S_{TZ}$. Then, using the sets $\{e(C_j,C_{j+1})\}_{j=1}^{t-1}\subseteq\{e(C,C')\;|\;\{C,C'\}\in S_{TZ}\}$ and $\{h(x),p(x)\}_{x\in V}$, we find a path $P$ in $G$ between $u,v$ using the algorithm from Lemma \ref{lemma:ExtractingAlgorithm2}. The resulting path $P$ has length of
\[
|P|\leq(|Q|+1)(2k+1)=t\cdot(2k+1)~,
\]
and it is returned as an output to the query.

Note that the path $Q$ that $TZ$ returned satisfies $|Q|=t-1\leq(2k-1)|R|$, where $R$ is the shortest path in $\mathcal{H}$ between $C(u)$ and $C(v)$. Let $P_{u,v}$ be the actual shortest path in $G$ between $u$ and $v$. Suppose that the vertices of $\mathcal{H}$ that $P_{u,v}$ passes through, by the order that it passes through them, are $(T_1,T_2,...T_q)$. By the definition of $\mathcal{H}$, there must be an edge $\{T_j,T_{j+1}\}$ in $\mathcal{H}$ for every $j\in[1,q-1]$. Hence, $R'=(T_1,T_2,...,T_q)$ is a path in $\mathcal{H}$, between $T_1=C(u)$ and $T_q=C(v)$, with length of at most $|P_{u,v}|=d_G(u,v)$. Since $R$ is the shortest path in $\mathcal{H}$ between $C(u)$ and $C(v)$, we have $|R|\leq d_G(u,v)$.

As a result,
\begin{eqnarray*}
|P|&\leq&t\cdot(2k+1)\\
&\leq&((2k-1)|R|+1)(2k+1)\\
&\leq&((2k-1)d_G(u,v)+1)(2k+1)\\
&=&(4k^2-1)d_G(u,v)+2k+1\\
&\leq&(4k^2+2k)d_G(u,v)=2k(2k+1)d_G(u,v)~.
\end{eqnarray*}
Thus, the stretch of our PRDO is at most $2k(2k+1)$.

The query time of our oracle consists of the time required for running a query of $TZ$, and of the time required for finding the path $P$. By Theorem \ref{thm:TZ} and Lemma \ref{lemma:ExtractingAlgorithm2}, the total time for these two computations is $O(\log k+|P|)$
which is $O(\log k)$ by our conventional PRDO notations.

As for the size of our PRDO, note that the variables $\{h(x),p(x)\}_{x\in V}$ (item $3$ in the description of the oracle) can be stored using only $O(n)$ space. The size of the set $\{e(C,C')\;|\;\{C,C'\}\in S_{TZ}\}$ equals to the size of $S_{TZ}$. Therefore, by Theorem \ref{thm:TZ}, the size of $TZ$, as well as the size of this set (items $1$ and $2$), is
\[
O(k|\mathcal{C}|^{1+\frac{1}{k}})~.
\]
Recall that by Lemma \ref{lemma:TreePartition}, the size of $\mathcal{C}$ is at most $\frac{n}{k}$. We conclude that the total size of our new PRDO is
\[
O(n+k\cdot(\frac{n}{k})^{1+\frac{1}{k}})=O(n^{1+\frac{1}{k}})~.
\]

\end{proof}

{\bf An Ultra-Compact PRDO for Unweighted Graphs.}
We can modify our PRDO for unweighted graphs, and get a PRDO of size $n+o(n)$. Here, the required storage for our PRDO is measured by \textit{words} - each of size at most $\log n$ bits. 
Decreasing the size of our PRDO is done at the cost of increasing the query time and (slightly) the stretch. The details are deferred to Appendix~\ref{sec:UltraCompactUnweighted}.

\section{Path-Reporting Distance Oracle for Weighted Graphs} \label{sec:NewPRDOSec}

In this section we devise our PRDO for weighted graphs. The basic framework is similar to the unweighted case: create a clustering of the graph, select a spanning tree for each cluster, and then apply the TZ PRDO over the cluster-graph. To answer a query $(u,v)$, we use the path in the cluster-graph between the clusters containing $u,v$, and complete it inside each cluster via the spanning trees edges.

The main differences from the unweighted setting are: 1) we use a more intricate clustering, {\em Bor\r{u}vka's clustering}, and 2) the trees spanning each cluster are not BFS trees, but are a subset of the MST (Minimum Spanning Tree) of the graph. These changes are needed in order to achieve the desired properties - that the number of clusters is small enough, while the stretch caused by going through the spanning trees of the clusters is controlled.


\subsection{Clustering via Bor\r{u}vka Forests} \label{sec:SpanningForests}

In this section we construct a clustering via a spanning forest of a graph. This construction is based on the well-known algorithm by Bor\r{u}vka for finding a minimum spanning tree in a graph. Similar constructions can be found in \cite{GH16,BBDHKLM17,BDGMN21}.

\begin{definition} \label{def:E_V}
Given an undirected weighted graph $G=(V,E)$, and a vertex $v\in V$, we denote by $e_v$ the minimum-weight edge among the adjacent edges to $v$ in the graph $G$. If there is more than one edge with this minimum weight, $e_v$ is chosen to be the one that is the smallest lexicographically.
\end{definition}

\begin{definition} \label{def:BoruvkaForest}
Given an undirected weighted graph $G=(V,E)$, the {\em Bor\r{u}vka Forest} of $G$ is the sub-graph $G'=(V,E')$ of $G$, where
\[
E'=\{e_v\;|\;v\in V\}~.
\]

Each connected component $T$ of $G'$ is called a {\em Bor\r{u}vka Tree}. The {\em root} of $T$ is chosen to be one of the adjacent vertices to the minimum-weight edge in $T$ (if there are several such minimum-weight edges, we pick the smallest one lexicographically, and the choice between its two adjacent vertices is arbitrary).
\end{definition}

To justify the use of the words "forest" and "tree", we prove the following lemma.
\begin{lemma} \label{lemma:BoruvkaForest}
The graph $G'$ is a forest. Moreover, if $T$ is a tree in $G'$, $x$ is a vertex of $T$, and $p(x)$ is $x$'s parent in $T$ (that is, the next vertex on the unique path from $x$ to the root of $T$), then $\{x,p(x)\}=e_x$.
\end{lemma}

\begin{proof}

First, we prove that $G'$ is a forest. Seeking contradiction, assume that $G'$ contains a cycle $C$, and let $\{u,v\}$ be the heaviest edge in $C$ (if there are several edges with the largest weight, choose the one that is largest lexicographically). Note that since $u$ has at least one adjacent edge in $C$, that is lighter than $\{u,v\}$, then it cannot be that $e_u=\{u,v\}$ (recall that $e_u$ is the lightest edge adjacent to $u$). Similarly, it cannot be that $e_v=\{u,v\}$. Hence, we get a contradiction to the fact that $\{u,v\}$ is an edge of $G'$ - since every such edge must be the edge $e_v$ of one of its endpoints $v$.

Next, Let $T$ be a tree in $G'$, denote its root by $v$, and let $x\neq v$ be a vertex of $T$. We prove by induction over the height of $x$, $h(x)$, which is the number of edges in the unique path between $x$ and $v$ in $T$.

When $h(x)=1$, we have $p(x)=v$. We consider two cases. If $\{x,v\}$ is the minimum-weight edge in $T$, then by definition $e_x$ must be this edge, i.e., $e_x=\{x,p(x)\}$. If $\{x,v\}$ is not the minimum-weight edge in $T$, then $e_v$ must be some other adjacent edge to $v$, thus $e_v\neq\{x,v\}$. But then, the reason that $\{x,v\}$ is in $E'$ must be that $\{x,p(x)\}=\{x,v\}=e_x$.

For $h(x)>1$, notice that $h(p(x))=h(x)-1$, and therefore by the induction hypothesis, $\{p(x),p(p(x))\}=e_{p(x)}$. But then, the edge $\{x,p(x)\}$ cannot be equal to $e_{p(x)}$, so it must be equal to $e_x$.

\end{proof}

The following lemma bounds the number of connected components (i.e., trees) in $G'$.
\begin{lemma} \label{lemma:NumberOfBoruvkaTrees}
The number of connected components in $G'$ is at most $\frac{1}{2}|V|$.
\end{lemma}

\begin{proof}
Let $C=(V_C,E'_C)$ be a connected component of $G'$, and let $x\in V_C$. The edge $e_x=\{x,y\}$ is in $C$, hence $y$ is also a vertex of $C$. In particular, $|V_C|\geq2$. Hence, if $\{C_i\}_{i=1}^t$ are the connected components of $G'$, then $|V_{C_i}|\geq2$ for every $i\in[1,t]$. Thus,
\[
\frac{1}{2}|V|=\frac{1}{2}\sum_{i=1}^t|V_{C_i}|\geq\frac{1}{2}\cdot2t=t~.
\]
\end{proof}

Next, we trim the trees in the Bor\r{u}vka forest so that each of them will be a \textit{star}, instead of a general tree. For this purpose, we will need the following definitions.

\begin{definition} \label{def:PartialBoruvkaForest}
Let $G'$ be the Bor\r{u}vka forest of $G$. For a vertex $x\in V$, denote by $h(x)$ the height of $x$ in the Bor\r{u}vka tree containing it. Define
\[
E'_0=\{\{a,b\}\in E'\;|\;\min\{h(a),h(b)\}=0\mod{2}\}~,
\]
\[
E'_1=\{\{a,b\}\in E'\;|\;\min\{h(a),h(b)\}=1\mod{2}\}~.
\]
We denote by $E''$ the largest set among these two.

Given an undirected weighted graph $G=(V,E)$, the {\em Partial Bor\r{u}vka Forest} of $G$ is the graph $G''=(V,E'')$.
\end{definition}

\begin{definition} \label{def:DirectedStar}
A {\em Star} is a rooted tree $S=(V_S,E_S)$ with root $v$ such that for every $x\in V_S\setminus\{v\}$, $\{x,v\}\in E_S$.
\end{definition}

\begin{lemma} \label{lemma:StarsForest}
The partial Bor\r{u}vka forest $G''=(V,E'')$ is a forest, where every tree is a star. In addition, if $S$ is a star in $G''$, $x$ is its root and $z\neq x$ is some other vertex of $S$, then $\{z,x\}=e_z$.
\end{lemma}

\begin{proof}

Notice that $G''$ is a sub-graph of the Bor\r{u}vka forest $G'$, hence $G''$ is also a forest. We assume that $E''=E'_0$, and the proof for the case where $E''=E'_1$ is symmetric.

Let $T$ be a tree in $G'$, with root $r$. 
Note that for any vertex $x\neq r$ we always have $h(x)=h(p(x))+1$, and thus $\min\{h(x),h(p(x))\}=h(p(x))=h(x)-1$. We conclude that if $h(x)$ is even, then $\{x,p(x)\}\notin E''$, and if $h(x)$ is odd, then $\{x,p(x)\}\in E''$.

Now let $S$ be a tree in $G''$, and let $x$ be the vertex in $S$ that has minimal $h(x)$. It cannot be that $h(x)$ is odd, otherwise $p(x)$ is connected to $x$ in $E''$, thus $p(x)$ is also in $S$ and has a smaller value of $h(p(x))=h(x)-1$. Therefore, $h(x)$ is even. By the discussion above we know that all of the children of $x$ in $T$ ($y$'s that satisfy $p(y)=x$) have an edge in $E''$ to $x$, but their children have no such edge. That is, $S$ is a star with $x$ as a root, where all the other vertices in $S$ are the children of $x$.

The last part of the lemma follows from the fact that we just proved, that the only other vertices in a star $S$ with a root $x$, are the children of $x$. By Lemma \ref{lemma:BoruvkaForest}, for every such child $z$, the edge $\{z,x\}=\{z,p(z)\}=e_z$.





\end{proof}

The following lemma bounds the number of trees in the partial Bor\r{u}vka forest of a graph.
\begin{lemma} \label{lemma:NumberOfStars}
The number of stars in $G''$ is at most $\frac{3}{4}|V|$.
\end{lemma}

\begin{proof}
Recall the Bor\r{u}vka forest $G'=(V,E')$. In every spanning forest $(V,F)$ of a graph $G=(V,E)$, the number of trees is exactly $|V|-|F|$. Thus, by Lemma \ref{lemma:NumberOfBoruvkaTrees}, we get
\[
|V|-|E'|\leq\frac{1}{2}|V|~,
\]
and therefore $|E'|\geq\frac{1}{2}|V|$. By the definition of $E''$, it contains at least half of these edges (since it equals to the larger set among two sets that cover the entire set $E'$). We conclude that $|E''|\geq\frac{1}{4}|V|$, and the number of trees in $G''$, which are stars, is
\[
|V|-|E''|\leq|V|-\frac{1}{4}|V|=\frac{3}{4}|V|~.
\]
\end{proof}

\subsubsection{A Hierarchy of Forests} \label{sec:HierarchyOfForests}



Given an undirected weighted graph $G=(V,E)$, we construct a sequence of forests $\{F_i=(V,E_i)\}_{i=0}^l$, where the integer parameter $l\geq0$ will be determined later. For $i=0$, define $E_0=\emptyset$. Then, for every $i\in[0,l]$, define the cluster-graph $\mathcal{H}_i=(\mathcal{C}_i,\mathcal{E}_i)$ as follows.

The set $\mathcal{C}_i$ is defined to be the set of the disjoint trees of the forest $F_i$. For every $T,T'\in\mathcal{C}_i$, denote by $e(T,T')$ the minimum-weight edge in $E$ among the edges between $T$ and $T'$. If there is no such edge, denote $e(T,T')=\bot$. Then define
\[
\mathcal{E}_i=\{\{T,T'\}\;|\;e(T,T')\neq\bot\}~.
\]
The weight of an edge $\{T,T'\}\in \mathcal{E}_i$ is defined to be the same as the weight of the edge $e(T,T')\in E$.

For any $i\in[0,l]$, given the graph $\mathcal{H}_i$, let $\mathcal{H}''_i=(\mathcal{C}_i,\mathcal{E}''_i)$ be the partial Bor\r{u}vka forest of $\mathcal{H}_i$. The graph $\mathcal{H}''_i$ is a disjoint union of stars. Let $\mathcal{S}$ be such star and let $T_0$ be its root.

Define the tree $Z$ in $G$ to be the tree that is formed by the union of the trees in $\mathcal{S}$ and the edges $e(T,T_0)$, for every $T\neq T_0$ in $\mathcal{S}$. The root of the tree $Z$ is defined to be the root of $T_0$. Finally, for any $i\in[0,l-1]$, the forest $F_{i+1}=(V,E_{i+1})$ is defined to be the disjoint union of the rooted trees $Z$ that are formed as was described, for all stars in $\mathcal{H}''_i$.

\begin{lemma} \label{lemma:NumberOfTrees}
For every $i\in[0,l]$, the forest $F_i$ has at most $(\frac{3}{4})^i|V|$ trees.
\end{lemma}

\begin{proof}

We prove the lemma by induction over $i\in[0,l]$. For $i=0$, $F_0$ is defined to be the graph $(V,\emptyset)$, so the number of trees in $F_0$ is $|V|$ and the claim holds.

For $i>0$, recall the graphs $\mathcal{H}_{i-1}$ and $\mathcal{H}''_{i-1}$ that were used in the definition of $F_i$. By the induction hypothesis, $F_{i-1}$ consists of at most $(\frac{3}{4})^{i-1}|V|$ trees, which are exactly the vertices of $\mathcal{H}_{i-1}$. Then, by Lemma \ref{lemma:NumberOfStars}, the graph $\mathcal{H}''_{i-1}$ has at most $\frac{3}{4}\cdot(\frac{3}{4})^{i-1}|V|=(\frac{3}{4})^i|V|$ stars. The forest $F_i$ consists of a single tree $Z$ for every star $\mathcal{S}$ in $\mathcal{H}''_{i-1}$, thus the number of trees in $F_i$ is at most $(\frac{3}{4})^i|V|$, as desired.

\end{proof}

\subsection{Stretch Analysis} \label{sec:StretchAnalysisWeighted}

Due to space considerations, we only state here the main lemma that will be used for bounding the stretch of our PRDO, without proof. The proof of this lemma, as well as some other lemmata , appears in Appendix \ref{sec:str-w}.

In the next lemma, we use the notations $p_i(x)$ and $h_i(x)$ to denote the parent and the height of $x$ in the tree of $F_i$ that contains $x$.

\begin{lemma} \label{lemma:ExtractingAlgorithm}
There is an algorithm that given two vertices $u,v\in V$, and a simple path $Q=(S_1,S_2,...,S_t)$ in the graph $\mathcal{H}_i$, such that $u$ is in $S_1$ and $v$ is in $S_t$, returns a path $P$ in $G$ between $u$ and $v$, with
\[
w(P)\leq3^{i+1}(d_G(u,v)+w(Q))~.
\]

The running time of the algorithm is proportional to the number of edges in the output path $P$. The required information for the algorithm is the set $\{h_i(x),p_i(x)\}_{x\in V}$, and the set $\{e(S_j,S_{j+1})\}_{j=1}^{t-1}$.
\end{lemma}

\subsection{A PRDO for Weighted Graphs} \label{sec:NewPRDOSub}

We are now ready to introduce our small size path-reporting distance oracle.

\begin{theorem} \label{thm:NewPRDO}
Let $G=(V,E)$ be an undirected weighted graph with $n$ vertices, and let $k\geq1$ be an integer parameter. There is a path-reporting distance oracle for $G$ with stretch $k^{\log_{4/3}4}<k^{4.82}$, query time $O(\log k)$ and size $O(n^{1+\frac{1}{k}})$.
\end{theorem}

\begin{proof}

Given the graph $G=(V,E)$, we construct the hierarchy of forests $\{F_i\}_{i=0}^l$ from Section \ref{sec:HierarchyOfForests}, where $l=\lfloor\log_{4/3}k\rfloor-2$. Consider the graph $\mathcal{H}_l=(\mathcal{C}_l,\mathcal{E}_l)$ that is defined in Section \ref{sec:HierarchyOfForests}. For every $x\in V$, denote by $h_l(x)$ the number of edges in the unique path from $x$ to the root of the tree of $F_l$ that $x$ belongs to. Let $p_l(x)$ be the parent of $x$ in that tree. Lastly, let $S(x)$ be the vertex of $\mathcal{H}_l$ (i.e., tree) that contains $x$.

Denote by $TZ$ the PRDO from Theorem \ref{thm:TZ} with the parameter $k$, when constructed over the graph $\mathcal{H}_l$. Let $S_{TZ}\subseteq\mathcal{E}_l$ be the set of edges of the underlying spanner of $TZ$.

Our new PRDO $D$ stores the following information.
\begin{enumerate}
    \item The oracle $TZ$.
    \item The set $\{e(T,T')\;|\;\{T,T'\}\in S_{TZ}\}$.
    \item The variables $\{h_l(x),p_l(x)\}_{x\in V}$. 
    \item The variables $\{S(x)\}_{x\in V}$.
\end{enumerate}

Given a query $(u,v)\in V^2$, the oracle $D$ queries $TZ$ on the vertices $S(u),S(v)$ of $\mathcal{H}_l$. Let $Q=(S(u)=S_1,S_2,...,S_t=S(v))$ be the resulting path, and note that all of its edges are in $S_{TZ}$. Then, using the sets $\{e(S_j,S_{j+1})\}_{j=1}^{t-1}\subseteq\{e(T,T')\;|\;\{T,T'\}\in S_{TZ}\}$ and $\{h_l(x),p_l(x)\}_{x\in V}$, the oracle $D$ uses the algorithm from Lemma \ref{lemma:ExtractingAlgorithm} to find a path $P$ in $G$ between $u,v$ with
\[
w(P)\leq3^{l+1}(d_G(u,v)+w(Q))~.
\]
The path $P$ is returned as an output to the query.
Note that the path $Q$ that $TZ$ returned satisfies
\[
w(Q)\leq(2k-1)w(R)~,
 \]
where $R$ is the shortest path in $\mathcal{H}_l$ between $S(u)$ and $S(v)$. Similarly to the proof of Theorem \ref{thm:NewPRDO}, it is easy to verify that $w(R)\leq d_G(u,v)$.

As a result,
\begin{eqnarray*}
w(P)&\leq&3^{l+1}(d_G(u,v)+w(Q))\\
&\leq&3^{l+1}(d_G(u,v)+(2k-1)w(R))\\
&\leq&3^{l+1}(d_G(u,v)+(2k-1)d_G(u,v))\\
&=&2k\cdot3^{l+1}d_G(u,v)\\
&\leq&2k\cdot3^{\log_{4/3}k-1}d_G(u,v)\\
&<&k^{1+\log_{4/3}3}d_G(u,v)=k^{\log_{4/3}4}d_G(u,v)~.
\end{eqnarray*}
Thus the stretch of our PRDO is smaller than $k^{\log_{4/3}4}$.

The query time of our oracle consists of the time required for running a query of $TZ$, and of the time required for computing the resulting path $P$ by Lemma \ref{lemma:ExtractingAlgorithm}. By Theorem \ref{thm:TZ} and Lemma \ref{lemma:ExtractingAlgorithm}, the total time for these two computations is $O(\log k+|P|)$,
which is $O(\log k)$ by our conventional PRDO notations.

As for the size of the PRDO $D$, note that the variables $\{h_l(x),p_l(x),S(x)\}_{x\in V}$ (items $3$ and $4$ in the description of $D$) can be stored using only $O(n)$ space. The size of the set $\{e(T,T')\;|\;\{T,T'\}\in S_{TZ}\}$ equals to the size of $S_{TZ}$. Therefore, by Theorem \ref{thm:TZ}, the size of $TZ$, as well as the size of this set (items $1$ and $2$ in the description of $D$), is
\[
O(k|\mathcal{C}_l|^{1+\frac{1}{k}})~.
\]
Recall that $\mathcal{C}_l$ is the set of vertices of $\mathcal{H}_l$. This set consists of the trees in the forest $F_l$. By Lemma \ref{lemma:NumberOfTrees}, the number of these trees is at most
\[
(\frac{3}{4})^l|V|=(\frac{3}{4})^{\lfloor\log_{4/3}k\rfloor-2}n\leq(\frac{3}{4})^{\log_{4/3}k-3}n=\frac{64n}{27k}~.
\]
Hence, the total size of our PRDO is
\[
O(n+k\cdot(\frac{64n}{27k})^{1+\frac{1}{k}})=O(n+n^{1+\frac{1}{k}})=O(n^{1+\frac{1}{k}})~.
\]

\end{proof}

{\bf An Ultra-Compact PRDO for Weighted Graphs.}
As in the unweighted version, the PRDO presented above can be fine-tuned into an ultra-compact PRDO (with size $n+o(n)$), at the cost of increasing the stretch and the query time. The details are deferred to Appendix~\ref{sec:ultra-compact-weighted}.


\section{Pairwise Path-Reporting Distance Oracle} \label{sec:PairwisePRDO}

Our construction of a pairwise PRDO relies on the pairwise spanner of Kogan and Parter, from their recent paper \cite{KP22} (in which the pairwise spanner is called a "near-exact preserver").
One of their useful results, that they also relied on for constructing their pairwise spanners, is the following lemma on hopsets. We first recall the definition of hopsets.

Let $G=(V,E)$ be a weighted undirected graph. For vertices $u,v\in V$ and some positive integer $\beta$, $d^{(\beta)}_G(u,v)$, denotes the weight of the lightest path between $u$ and $v$ in $G$, among the paths that have at most $\beta$ edges. An {\em $(\alpha,\beta)$-hopset} is a set $H\subseteq\binom{V}{2}$, such that for every two vertices $u,v\in V$,
\[
d_G(u,v)\le d^{(\beta)}_{G\cup H}(u,v)\leq\alpha\cdot d_G(u,v)~,
\]
where the weight of an edge $(x,y)\in H$ is defined to be $d_G(x,y)$.

The proof of the following lemma can be found in \cite{KP22}.

\begin{lemma}[Lemma 4.4 from \cite{KP22}] \label{lemma:HopsetsTradeoff}
Let $G=(V,E)$ be an undirected weighted graph on $n$ vertices, and let $k,D\geq1$ be integer parameters. For every $0<\epsilon<1$, there exists a $(1+\epsilon,\beta)$-hopset $H$ for $G$, where $\beta=O(\frac{\log k}{\epsilon})^{\log k}\cdot D$ and
\[
|H|=O\left(\left(\frac{n\log n}{D}\right)^{1+\frac{1}{k}}\right)~.
\]
\end{lemma}

Similarly to the constructions in \cite{KP22}, we now show how a pairwise PRDO can be produced, using the hopsets from Lemma \ref{lemma:HopsetsTradeoff}. We will use the notation $\beta(\epsilon,k)=O(\frac{\log k}{\epsilon})^{\log k}$ for brevity.

\begin{theorem} \label{thm:PairwisePRDO}
Let $G=(V,E)$ be an undirected weighted graph on $n$ vertices and let $\mathcal{P}\subseteq V^2$ be a set of pairs of vertices. For every $\epsilon\in(0,1)$, there exists a pairwise path-reporting distance oracle with stretch $1+\epsilon$, query time $O(1)$ and size
\[
O\left(\frac{\log n\cdot(\log\log n)^2}{\epsilon}\right)^{\log\log n}\cdot\tilde{O}\big(|\mathcal{P}|+n\big)=n^{o_\epsilon(1)}\cdot O(n+|\mathcal{P}|)~.
\]
\end{theorem}

\begin{proof}

Let $n=D_0>D_1>\cdots>D_l=2$ be some sequence of integer parameters that will be determined later. Denote $k=\log n$, and for a given $\epsilon\in(0,1)$, denote $\epsilon'=\frac{\epsilon}{2(l+1)}$. Let $H_0,H_1,...,H_l$ be the resulting hopsets when applying Lemma \ref{lemma:HopsetsTradeoff} on $\epsilon',k=\log n$ and $D_0,D_1,...,D_l$ respectively. That is, $H_i$ is a $(1+\epsilon',\beta_i)$-hopset with size $O((\frac{n\log n}{D_i})^{1+\frac{1}{k}})$, where $\beta_i=\beta(\epsilon',k)\cdot D_i$. For $i=0$, note that $\beta_i\geq n$, thus we can simply assume that $H_0=\emptyset$ (if it is not the case, we \textit{define} $H_0$ to be $\emptyset$, which is a $(1,n)$-hopset).

We now define our oracle $D$ to contain the following information. For every $i\in[1,l]$ and for every $(x,y)\in H_i$, let $Q_{x,y}$ be the shortest path in $G\cup H_{i-1}$ between $x,y$, among the paths that contain at most $\beta_{i-1}$ edges. In addition, for every $(x,y)\in\mathcal{P}$, let $P_{x,y}$ be the shortest path in $G\cup H_l$ between $x,y$, among the paths with at most $\beta_l$ edges. Our oracle $D$ stores all of these paths: $\bigcup_{i=1}^l\{Q_{x,y}\}_{(x,y)\in H_i}\cup\{P_{x,y}\}_{(x,y)\in\mathcal{P}}$.

Given a query $(u,v)\in\mathcal{P}$, we find the path $P_l=P_{u,v}\subseteq G\cup H_l$ that is stored in $D$. Then, we replace every edge $(x,y)\in H_l$ on $P_l$ by the corresponding path $Q_{x,y}\subseteq G\cup H_{l-1}$. The result is a path $P_{l-1}$ between $u,v$ in $G\cup H_{l-1}$. Every edge $(x,y)\in H_{l-1}$ on $P_{l-1}$ is then replaced by the path $Q_{x,y}\subseteq G\cup H_{l-2}$, to get a path $P_{l-2}$ between $u,v$ in $G\cup H_{l-2}$. We continue in the same way, until finally reaching to a path $P_0$ between $u,v$ in the graph $G\cup H_0=G$. We return $P_0$ as an output to the query.

By the hopset property, we know that
\[
w(P_l)=w(P_{u,v})=d^{(\beta_l)}_{G\cup H_l}(u,v)\leq(1+\epsilon')d_G(u,v)~.
\]
Similarly, every $(x,y)\in P_l$ that is also in $H_l$, is replaced with the path $Q_{x,y}$, that has a weight of
\[
w(Q_{x,y})=d^{(\beta_{l-1})}_{G\cup H_{l-1}}(x,y)\leq(1+\epsilon')d_G(x,y)=(1+\epsilon')w(x,y)~.
\]
Thus, the resulting path $P_{l-1}$ has a weight of at most $1+\epsilon'$ times the weight of $P_l$, that is
\[
w(P_{l-1})\leq(1+\epsilon')w(P_l)\leq(1+\epsilon')^2d_G(u,v)~.
\]
Proceeding in the same way, we conclude that $w(P_0)\leq(1+\epsilon')^{l+1}d_G(u,v)$.
Hence, the stretch of our distance oracle is
\[
(1+\epsilon')^{l+1}=(1+\frac{\epsilon}{2(l+1)})^{l+1}\leq e^{\frac{\epsilon}{2}}\leq1+\epsilon~.
\]

For analysing the query time of our distance oracle, we can think of the query algorithm as a single pass on the path $P_l$, where every time that an edge of $H_l$ is reached, we replace it with the appropriate path $Q_{x,y}$, and continue inside $Q_{x,y}$ recursively. Since every step produces an edge that will appear in the output path, the query time is proportional to this output path. Observe, however, that the resulting path is actually a walk, and not necessarily a simple path. By our convention of writing the query time of PRDOs, this query time is $O(1)$.

Lastly, we analyse the size of our pairwise PRDO. Note that by their definitions, the paths $P_{x,y}$, for every $(x,y)\in\mathcal{P}$ are of length at most $\beta_l$. Similarly, the length of $Q_{x,y}$, for $(x,y)\in H_i$ is at most $\beta_{i-1}$. Therefore, the total space required for storing these paths is at most
\begin{eqnarray*}
|\mathcal{P}|\cdot\beta_l+\sum_{i=1}^l|H_i|\cdot\beta_{i-1}&=&|\mathcal{P}|\cdot\beta(\epsilon',k)\cdot D_l+\sum_{i=1}^lO\left(\left(\frac{n\log n}{D_i}\right)^{1+\frac{1}{k}}\right)\cdot\beta(\epsilon',k)\cdot D_{i-1}\\
&=&\beta(\epsilon',k)\cdot\left(|\mathcal{P}|\cdot2+\sum_{i=1}^lO\left(\left(\frac{n\log n}{D_i}\right)^{1+\frac{1}{k}}\right)\cdot D_{i-1}\right)\\
&=&\beta(\epsilon',k)\cdot O\big(|\mathcal{P}|+(n\log n)^{1+\frac{1}{k}}\sum_{i=1}^l\frac{D_{i-1}}{D_i^{1+\frac{1}{k}}}\big)\\
&=&O\left(\frac{\log k}{\epsilon/2l}\right)^{\log k}\cdot O\big(|\mathcal{P}|+(n\log n)^{1+\frac{1}{\log n}}\sum_{i=1}^l\frac{D_{i-1}}{D_i^{1+\frac{1}{k}}}\big)\\
&=&O\left(\frac{l\cdot\log k}{\epsilon}\right)^{\log k}\cdot O\big(|\mathcal{P}|+n\log n\cdot\sum_{i=1}^l\frac{D_{i-1}}{D_i^{1+\frac{1}{k}}}\big)\\
&=&O\left(\frac{l\cdot\log k}{\epsilon}\right)^{\log k}\cdot\tilde{O}\big(|\mathcal{P}|+n\cdot\sum_{i=1}^l\frac{D_{i-1}}{D_i^{1+\frac{1}{k}}}\big)~.
\end{eqnarray*}

For making the last term small, we choose $D_i=\left\lceil n^{(\frac{k}{k+1})^i}\right\rceil$, and thus\newline
$
\frac{D_{i-1}}{D_i^{1+\frac{1}{k}}}\leq\frac{n^{(\frac{k}{k+1})^{i-1}}+1}{n^{(\frac{k}{k+1})^i\cdot(1+\frac{1}{k})}}=\frac{n^{(\frac{k}{k+1})^{i-1}}+1}{n^{(\frac{k}{k+1})^{i-1}}}\leq2~.
$
For this choice of $D_i$, since we want $D_l$ to be $2$, we must have $n^{(\frac{k}{k+1})^l}\leq2$, that is, $l\geq\log_{\frac{k+1}{k}}(\log n)$. Notice that
$
\log_{\frac{k+1}{k}}(\log n)=\frac{\log\log n}{\log(1+\frac{1}{k})}\leq\frac{\log\log n}{\log(2^\frac{1}{k})}=k\log\log n~,
$
thus we can choose $l=\lceil k\log\log n\rceil=\lceil\log n\cdot\log\log n\rceil$.

In conclusion, the size of our pairwise PRDO is at most
\begin{eqnarray*}
O\left(\frac{l\cdot\log k}{\epsilon}\right)^{\log k}\cdot\tilde{O}\big(|\mathcal{P}|+n\cdot\sum_{i=1}^l\frac{D_{i-1}}{D_i^{1+\frac{1}{k}}}\big)&=&O\left(\frac{l\cdot\log k}{\epsilon}\right)^{\log k}\cdot\tilde{O}\big(|\mathcal{P}|+n\cdot\sum_{i=1}^l2\big)\\
&=&O\left(\frac{l\cdot\log k}{\epsilon}\right)^{\log k}\cdot\tilde{O}\big(|\mathcal{P}|+l\cdot n\big)\\
&=&O\left(\frac{\log n\cdot(\log\log n)^2}{\epsilon}\right)^{\log\log n}\cdot\tilde{O}\big(|\mathcal{P}|+n\big)\\
&=&n^{o_\epsilon(1)}\cdot O\big(|\mathcal{P}|+n\big)\\
\end{eqnarray*}

\end{proof}

{\bf Acknowledgement.} The second-name author would like to thank Michael Elkin, who also suggested that the pairwise spanner of \cite{KP22} might be extended to a pairwise PRDO, and that a PRDO for unweighted graphs with size $O\left(n^{1+\frac{1}{k}}\right)$ and stretch $O(k^2)$ can be achieved.

\bibliography{hopset}

\newpage
\appendix

\section{An Ultra-Compact PRDO for Unweighted Graphs} \label{sec:UltraCompactUnweighted}
\begin{theorem} \label{thm:UltraCompactPRDO}
Let $G=(V,E)$ be an undirected unweighted graph with $n$ vertices, and let $k\in[1,\log n]$ be some integer parameter. For every $\epsilon\in[\frac{1}{\log n},\frac{k}{\log k}]$, There is a path-reporting distance oracle for $G$ with stretch $O(\frac{k^2}{\epsilon})$, query time $O(\frac{k}{\epsilon})$ and size $n+\epsilon n^{1+\frac{1}{k}}+o(n)$.

In particular, for $k=\log n$ and $\epsilon=o(1)$, we get a PRDO with stretch $O(\frac{\log^2n}{\epsilon})$, query time $O(\frac{\log n}{\epsilon})$ and size $n+o(n)$.
\end{theorem}

\begin{proof}

Recall that by Theorem \ref{thm:TZ}, every graph with $n$ vertices has a PRDO of size $O(kn^{1+\frac{1}{k}})$, query time $O(\log k)$ and stretch $2k-1$. Let $c$ be the constant such that this PRDO is always of size at most
\begin{equation} \label{eq:TZSize}
    ckn^{1+\frac{1}{k}}~.
\end{equation}
Let $\mathcal{C}$ be the clustering from Lemma \ref{lemma:TreePartition}, but with radius $\frac{2ck}{\epsilon}$ instead of $k$. That is, for every cluster $C\in\mathcal{C}$ there is a spanning tree $T[C]$ with a root $r[C]$, such that
\[
\forall_{v\in C}\;d_{T[C]}(v,r[C])\leq\frac{2ck}{\epsilon}~.
\]
In addition, the number of clusters in $\mathcal{C}$ is at most $\frac{n}{2ck/\epsilon}=\frac{\epsilon n}{2ck}$.

We repeat the same construction as the one in Theorem \ref{thm:NewUnweightedPRDO}: Let $TZ$ be the oracle from Theorem \ref{thm:TZ} with parameter $k$, when constructed over $\mathcal{H}=(\mathcal{C},\mathcal{E})$ (see the definition of $\mathcal{H}$ in Definition \ref{def:ClustersGraph}). Let $S_{TZ}$ be its underlying spanner. Our ultra-compact PRDO stores $TZ$, the set $\{e(C,C')\;|\;(C,C')\in S_{TZ}\}$ and the variables $\{p(x),h(x)\}_{x\in V}$.

Note that in contrary to our construction from Theorem \ref{thm:NewUnweightedPRDO}, here we do not store $C(x)$, the cluster of $\mathcal{C}$ that contains $x$, for every $x\in V$. Thus, we have to modify the query algorithm so that it will not use this information. Given $(u,v)\in V^2$, the query algorithm first follows the pointers $p()$ from $u$ and from $v$ until it reaches the roots of the spanning trees of $C(u)$ and $C(v)$. Without loss of generality, we assume that these roots are the representatives of the corresponding clusters in the graph $\mathcal{H}$. We then proceed in the same way as the original query algorithm.

Due to the change in the query algorithm, we now have an additional term of \newline$O(d(u,r[C(u)])+d(v,r[C(v]))$ in the query time. Since $\mathcal{C}$ is a clustering of radius $\frac{2ck}{\epsilon}$, this term is $O(\frac{2ck}{\epsilon})=O(\frac{k}{\epsilon})$. Besides that, the rest of the query algorithm is not affected by the changes in the PRDO, thus has the same running time of $O(\log k)$. We get a total query time of $O(\log k+\frac{k}{\epsilon})=O(\frac{k}{\epsilon})$, since $\epsilon\leq\frac{k}{\log k}$.

Note that since we replaced $k$ by $\frac{2ck}{\epsilon}$ as the radius of the clustering $\mathcal{C}$, we get in Lemma \ref{lemma:ExtractingAlgorithm2} that given a path $(C_1,...,C_t)$ in $\mathcal{H}$, the corresponding path $P$ in $G$ satisfies
\[
|P|\leq t\cdot(\frac{4ck}{\epsilon}+1)~.
\]
See Lemma \ref{lemma:ExtractingAlgorithm2} for more details. Therefore, following the same proof as of Theorem \ref{thm:NewUnweightedPRDO}, we conclude that the stretch of this new PRDO is
\[
O(k\cdot\frac{2ck}{\epsilon})=O(\frac{k^2}{\epsilon})~.
\]

Finally, we compute the size of our new PRDO more carefully. First, by (\ref{eq:TZSize}), $TZ$ and the underlying spanner $S_{TZ}$ each have size at most
\[
c\cdot k\cdot\left(\frac{\epsilon n}{2ck}\right)^{1+\frac{1}{k}}\leq\frac{1}{2}\epsilon n^{1+\frac{1}{k}}~.
\]
Hence, their total size is at most $\epsilon n^{1+\frac{1}{k}}$.

Next, notice that $h(x)$, for every $x\in V$ is the distance from $x$ to the root of $C(x)$. By the limitation on the radius of the cluster $C(x)$, we know that this distance $h(x)$ is always at most $\frac{2ck}{\epsilon}$, thus storing it takes $\log\frac{2ck}{\epsilon}=\log\frac{k}{\epsilon}+O(1)$ \textit{bits} of space. Since we measure the storage size by \textit{words}, we divide by $\log n$:
\[
\frac{1}{\log n}\cdot\left(\log\frac{k}{\epsilon}+O(1)\right)\leq\frac{1}{\log n}\cdot\left(\log\left(\frac{\log n}{\epsilon}\right)+O(1)\right)=o(1)~.
\]
In the last step we used the fact that $\epsilon\geq\frac{1}{\log n}$.

Lastly, notice that storing the variables $\{p(x)\}_{x\in V}$ requires storing a single word for every vertex $x\in V$. We conclude that the total number of words in our PRDO is
\[
\epsilon n^{1+\frac{1}{k}}+\sum_{x\in V}1+\sum_{x\in V}o(1)=n+\epsilon n^{1+\frac{1}{k}}+o(n)~,
\]
as desired.

\end{proof}

\section{An Ultra-Compact PRDO for Weighted Graphs}\label{sec:ultra-compact-weighted}

We start with the following lemma.

\begin{lemma} \label{lemma:UnweightedRadius}
For every choice of the parameter $l$, the trees of $\mathcal{C}_l$ have an {\em unweighted} radius at most $\frac{3^l-1}{2}$. That is, for every tree $S\in\mathcal{C}_l$ with root $r$, and for every $v$ in $S$, the length (number of edges) in the unique path from $v$ to $r$ in $S$ is at most $\frac{3^l-1}{2}$.
\end{lemma}

\begin{proof}

We prove the lemma by induction on $l$. For $l=0$, the trees of $\mathcal{C}_0$ are single vertices, thus have a radius of $0=\frac{3^0-1}{2}$. For $l>0$, let $S\in\mathcal{C}_l$ be a tree with root $r$. The tree $S$ is the union of a tree $T_0\in\mathcal{C}_{l-1}$ with some other trees $T'\in\mathcal{C}_{l-1}$, together with an edge that connects them to $T_0$. The root $r$ is defined to be the root of $T_0$. Given $v\in S$, suppose that $v\in T'\in\mathcal{C}_{l-1}$. The unique path from $v$ to $r$ goes from $v$, to the edge connecting $T'$ and $T_0$, then from this edge to the root of $T_0$, which is $r$. The number of edges until reaching the connecting edge is at most the unweighted \textit{diameter} of $T'$, which is bounded by $2\cdot\frac{3^{l-1}-1}{2}=3^{l-1}-1$, using the induction hypothesis. The number of edges after this edge is at most the radius of $T_0$, which is at most $\frac{3^{l-1}-1}{2}$, again by the induction hypothesis. We finally get that the total number of edges in this unique path from $v$ to $r$ is at most
\[
3^{l-1}-1+1+\frac{3^{l-1}-1}{2}=\frac{3}{2}(3^{l-1}-1)+1=\frac{3^l-1}{2}~.
\]

\end{proof}

We are now ready to prove the result about ultra-compact PRDOs:
\begin{theorem}
Let $G=(V,E)$ be an undirected weighted graph with $n$ vertices, and let $k\geq1$ be an integer parameter. For every $\epsilon\in[\frac{1}{\log n},\frac{k}{\log k}]$, There is a path-reporting distance oracle for $G$ with stretch $O(k(\frac{k}{\epsilon})^{\log_{4/3}3})$, query time $O((\frac{k}{\epsilon})^{\log_{4/3}3})$ and size $n+\epsilon n^{1+\frac{1}{k}}+o(n)$.

In particular, for $k=\log n$ and $\epsilon=o(1)$ (say $\epsilon=k^{-0.0001}$, or even $\epsilon=\frac{1}{\log^*n}$), we get a PRDO with size $n+o(n)$, whose stretch and query time are at most $O((\frac{k}{\epsilon})^{\log_{4/3}4})=O(k^{4.82})=O(\log^{4.82}n)$.
\end{theorem}

\begin{proof}

Let $c$ be the constant from the O-notation in the size of the TZ PRDO from Theorem \ref{thm:TZ}. That is, the size of this PRDO is at most $ckn^{1+\frac{1}{k}}$. We choose $l=\lceil\log_{4/3}\frac{k}{\epsilon}+\log_{4/3}(2c)\rceil$, and compute the graph $\mathcal{H}_l=(\mathcal{C}_l,\mathcal{E}_l)$.

Our ultra-compact PRDO will store the exact same information as the PRDO in the proof of Theorem \ref{thm:NewPRDO}, except for the variables $\{S(x)\}_{x\in V}$, where $S(x)$ is the tree of $\mathcal{C}_l$ that contains $x$. Accordingly, we change the query algorithm as follows. Given $(u,v)\in V^2$, we first find $S(u),S(v)$ by following the pointers $p()$ from $u$ and from $v$. We assume that the roots of these trees represent the trees themselves. The rest of the algorithm is identical to the one in Theorem \ref{thm:NewPRDO}.

The additional query time, that was added due to the first step of the query algorithm, is proportional to the unweighted radius of the trees $S(u),S(v)$. I.e., it is proportional to the length of the unique path from $u,v$ to the roots of their trees. In Lemma \ref{lemma:UnweightedRadius} we proved that this length is at most
\[
\frac{3^l-1}{2}=O(3^l)=O(3^{\log_{4/3}\frac{k}{\epsilon}})=O((\frac{k}{\epsilon})^{\log_{4/3}3})~.
\]
Since $\epsilon\leq\frac{k}{\log k}$, this term is more significant than the $O(\log k)$ term from the original analysis, thus the total query time is indeed $O((\frac{k}{\epsilon})^{\log_{4/3}3})$.

By an identical proof to the one of Theorem \ref{thm:NewPRDO}, we can show that the stretch of our PRDO is
\[
2k\cdot3^{l+1}=O(k\cdot3^l)=O(k(\frac{k}{\epsilon})^{\log_{4/3}3})~.
\]

Regarding the size of our PRDO, we start by bounding the size of both the PRDO $TZ$ and its underlying spanner $S_{TZ}$ by
\[
c\cdot k|\mathcal{C}_l|^{1+\frac{1}{k}}\leq ck\left((\frac{3}{4})^ln\right)^{1+\frac{1}{k}}\leq ck\left(\frac{n}{2ck/\epsilon}\right)^{1+\frac{1}{k}}\leq\frac{\epsilon n^{1+\frac{1}{k}}}{2}~.
\]
We used here Lemma \ref{lemma:NumberOfTrees} to bound $|\mathcal{C}_l|$. Therefore, the total size of these two components is at most $\epsilon n^{1+\frac{1}{k}}$.

As for the variables $h_l()$, notice that they are all numbers that are bounded by the unweighted radius of the trees of $\mathcal{C}_l$, which is at most $\frac{3^l-1}{2}$ by Lemma \ref{lemma:UnweightedRadius}. Thus, the total number of \textit{words} each $h_l(x)$ requires is
\[
\frac{1}{\log n}\cdot\log\frac{3^l-1}{2}<\frac{\log(3^l)}{\log n}=\frac{(\log_{4/3}\frac{k}{\epsilon}+O(1))\log3}{\log n}\leq\frac{O(\log\log n)}{\log n}=o(1)~,
\]
where we used the fact that $\epsilon\geq\frac{1}{\log n}$.

We conclude that the number of words that are required for our PRDO (which includes also the single word $p(x)$ for each vertex $x\in V$) is
\[
\epsilon n^{1+\frac{1}{k}}+\sum_{x\in V}1+\sum_{x\in V}o(1)=n+\epsilon n^{1+\frac{1}{k}}+o(n)~.
\]

\end{proof}

\section{Stretch Analysis for Weighted Graphs}\label{sec:str-w}

We now state and prove several lemmata that will be useful for bounding the stretch of our PRDO.

\begin{lemma} \label{lemma:ExternalEdges}
Let $\{a,b\},\{c,d\}\in E$ be edges in $G$ such that $a$ and $c$ are in the same tree $T$ of the forest $F_i$, while $b,d$ are not in $T$. Let $v$ be the root of $T$ and let $P_{a,v}$ and $P_{v,c}$ be the unique paths in $T$ from $a$ to $v$ and from $v$ to $c$ respectively. Then, the path $P=P_{a,v}\circ P_{v,c}$ in $G$ satisfies
\[
w(P)\leq\frac{3^i-1}{2}(w(a,b)+w(c,d))~.
\]
\end{lemma}

\begin{proof}

The proof is by induction over $i\in[0,l]$. For $i=0$, since the trees of $F_0$ are of size $1$, it must be that $a=c=v$. Thus $P$ is an empty path, and the claim holds.

For $i>0$, the tree $T$ was created by a star $\mathcal{S}$ in $\mathcal{H}''_{i-1}$. Note that while $T$ and $\mathcal{S}$ conceptually represent the same set of vertices in the original graph $G$, the tree $T$ is actually a sub-graph of $G$, while $\mathcal{S}$ is a sub-graph of $\mathcal{H}''_{i-1}$, whose vertices are trees in $G$. The star $\mathcal{S}$ consists of a tree $T_0$ of $F_{i-1}$, and some other trees $T'$ of $F_{i-1}$ such that $\{T',T_0\}\in\mathcal{E}''_{i-1}$.

Consider the paths $P_{a,v},P_{v,c}$. The root $v$ of $T$ was chosen to be the root of $T_0$, hence both paths must pass through $T_0$. Let $a_0$ be the first vertex of $P_{a,v}$, in the direction from $a$ to $v$, that belongs to $T_0$. Similarly, let $c_0$ be the last vertex of $P_{v,c}$, in the direction from $v$ to $c$, that belongs to $T_0$.

We first assume that $a_0\neq a$. Let $a'$ be the preceding vertex to $a_0$ on $P_{a,v}$. The sub-path $P_{a,a'}$, from $a$ to $a'$, must be contained in a tree $T'_a$ of $F_{i-1}$, where $e(T'_a,T_0)=\{a',a_0\}$. Let $v'$ be the root of $T'_a$ (see Figure \ref{fig:ExternalEdges} for an illustration). Then the path $P_{a,a'}$ is contained in the concatenation of two paths $P_{a,v'}$ and $P_{v',a'}$ in $T'_a$, from $a$ to $v'$ and from $v'$ to $a'$ respectively. By the induction hypothesis,
\[
w(P_{a,a'})\leq w(P_{a,v'}\circ P_{v',a'})\leq\frac{3^{i-1}-1}{2}(w(a,b)+w(a',a_0))~.
\]

\begin{center}
\begin{figure}[ht!]
    \centering
    \includegraphics[width=9.5cm, height=5cm]{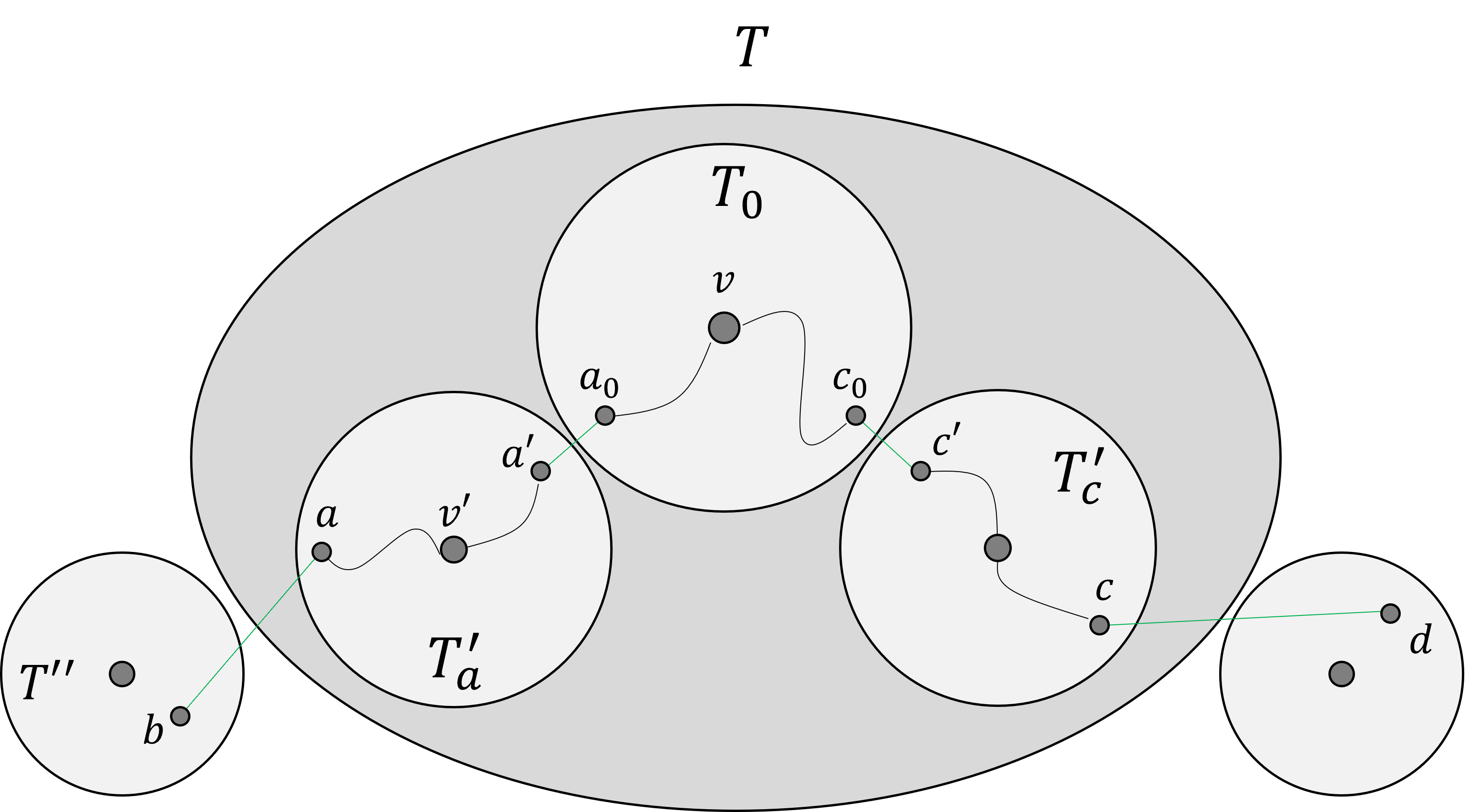}
    \caption{An illustration of Lemma \ref{lemma:ExternalEdges}. Straight green lines represent edges of the graph $G$, while black curvy lines represent paths.}
    \label{fig:ExternalEdges}
\end{figure}
\end{center}

Recall that by Lemma \ref{lemma:StarsForest}, we have $\{T'_a,T_0\}=e_{T'_a}$. That is, the weight of $\{T'_a,T_0\}$ is smaller or equal to the weight of every adjacent edge to $T'_a$ in $\mathcal{H}_{i-1}$. In particular, if $T''$ is the vertex of $\mathcal{H}_{i-1}$ that contains $b$ (it cannot be that $T''=T'_a$, as $b$ is not in $T$, hence not in $T'_a$), then
\[
w(a',a_0)=w(e(T'_a,T_0))=w(T'_a,T_0)=w(e_{T'_a})\leq w(T'_a,T'')=w(e(T'_a,T''))\leq w(a,b)~,
\]
where the last step is due to the definition of $e(\cdot,\cdot)$. We finally conclude that
\[
w(P_{a,a'})\leq\frac{3^{i-1}-1}{2}(w(a,b)+w(a',a_0))\leq(3^{i-1}-1)w(a,b)~.
\]

Symmetrically, we conclude that if $c_0\neq c$, $c'$ is the next vertex after $c_0$ on $P_{v,c}$, and $P_{c',c}$ is the sub-path of $P_{v,c}$ from $c'$ to $c$, then
\[
w(c_0,c')\leq w(c,d),\;\;\;\;\;\text{and}\;\;\;\;\;w(P_{c',c})\leq(3^{i-1}-1)w(c,d)~.
\]

Let $P_{a_0,v},P_{v,c_0}$ be the sub-paths of $P_{a,v},P_{v,c}$ from $a_0$ to $v$ and from $v$ to $c_0$ respectively. Using the induction hypothesis, we get
\[
w(P_{a_0,v}\circ P_{v,c_0})\leq\frac{3^{i-1}-1}{2}(w(a_0,a')+w(c_0,c'))\leq\frac{3^{i-1}-1}{2}(w(a,b)+w(c,d))~.
\]


The other cases where $a_0=a$ and $c_0\neq c$, or symmetrically $a_0\neq a$ and $c_0=c$, and the case where $a_0=a,c_0=c$, can be all analyzed in the same way. In every possible case, we get to the conclusion that
\[
w(P_{a_0,v}\circ P_{v,c_0})\leq\frac{3^{i-1}-1}{2}(w(a,b)+w(c,d))~.
\]
In addition, whenever the vertices $a'$ and $c'$ are defined,
\[
w(P_{a,a'})+w(a',a_0)\leq(3^{i-1}-1)w(a,b)+w(a,b)=3^{i-1}w(a,b)~,
\]
and
\[
w(c_0,c')+w(P_{c',c})\leq w(c,d)+(3^{i-1}-1)w(c,d)=3^{i-1}w(c,d)~.
\]
The path $P=P_{a,v}\circ P_{v,c}$ is contained in the concatenation of the above paths and edges; the path $P_{a_0,v}\circ P_{v,c_0}$, and the paths $P_{a,a'}$, $\{a',a_0\}$, $P_{c,c'}$, $\{c',c_0\}$ whenever they are defined.
Therefore we finally get
\[
w(P)\leq\frac{3^{i-1}-1}{2}(w(a,b)+w(c,d))+3^{i-1}(w(a,b)+w(c,d))=\frac{3^i-1}{2}(w(a,b)+w(c,d))~.
\]

\end{proof}

The proofs of the next two lemmata follow the same lines.

\begin{lemma} \label{lemma:InternalEdges}
Let $\{a,b\}\in E$ be an edge of $G$, such that $a,b$ are in the same tree $T$ of $F_i$, for some $i\in[0,l]$. Let $z=lca(a,b)$ be the lowest common ancestor in $T$ (see Section \ref{sec:Preliminaries} for definition) and let $P_{a,z},P_{b,z}$ be the unique paths in $T$ from $a$ to $z$ and from $b$ to $z$ respectively. Then,
\[
w(P_{a,z})+w(P_{b,z})\leq(3^i-1)w(a,b)~.
\]
\end{lemma}

\begin{proof}

First note that it is enough to prove this lemma for the \textit{minimal} index $i\in[0,l]$ such that $a,b$ are in the same tree of $F_i$. This is true since if $i$ is this minimal index, then $a,b$ belong to the same tree $T'$ of $F_j$ if and only if $i\leq j$ (the forests $F_j$ are a non-decreasing series of sets), and $P_{a,z},P_{b,z}$ together have the same edges in $T'$ as they have in $T$. Therefore
\[
w(P_{a,z})+w(P_{b,z})\leq(3^i-1)w(a,b)\leq(3^j-1)w(a,b)~.
\]

Indeed, let $i$ be this minimal index, and let $T$ be the tree of $F_i$ that contains both $a$ and $b$. If $i=0$, then $T$ contains a single vertex, and thus $a=b$. Hence, the paths $P_{a,z},P_{b,z}$ are empty and the claim is trivial.

Suppose that $i>0$. By the construction of $F_i$, we know that there is a star $\mathcal{S}$ in $\mathcal{H}''_{i-1}$, such that $T$ is the union of the trees in $\mathcal{S}$. The star $\mathcal{S}$ consists of a tree $T_0$ of $F_{i-1}$ and some other trees $T'$ of $F_{i-1}$ such that $\{T',T_0\}=e_{T'}$ is an edge of $\mathcal{H}_{i-1}$.

We consider the path $P$ to be the concatenation of the paths $P_{a,z}$ and $P_{z,b}$ ($P_{z,b}$ is the same path as $P_{b,z}$). Note that $P$ cannot be contained in a single tree $T_0$ or $T'$, since it would contradict the minimality of $i$. Therefore, $P$ must pass through $T_0$, as the only edges in $T$ connecting two different trees of $\mathcal{S}$ are from some $T'$ to $T_0$. From the same reason, if $P_{a_0,b_0}$ is the sub-path of $P$ that is in $T_0$, then each of the sub-paths $P_{a,a_0}$ (from $a$ to $a_0$) and $P_{b,b_0}$ (from $b$ to $b_0$) must be contained in two different trees $T',T''$ of $\mathcal{S}$, except for their last edge. Note that at least one of these two sub-paths is not empty, as otherwise $P$ would be contained in $T_0$.

If $P_{a,a_0}$ is not empty, we can write $P_{a,a_0}=P_{a,a'}\circ\{a',a_0\}$ where $P_{a,a'}$ is in $T'$, $a_0$ is in $T_0\neq T'$, and $\{a',a_0\}=e(T',T_0)$. Also, $\{T',T_0\}=e_{T'}$ in the graph $\mathcal{H}_{i-1}$. The vertex $b$ cannot be in $T'$, otherwise both $a,b$ would be contained in $T'$, again in contradiction to the minimality of $i$. Thus, $b$ is in some tree $T''\neq T'$ of $\mathcal{S}$. By the definitions of $e_{T'}$ and $e(\cdot,\cdot)$,
\[
w(a',a_0)=w(e(T',T_0))=w(T',T_0)=w(e_{T'})\leq w(T',T'')=w(e(T',T''))\leq w(a,b)~.
\]
The edges of the path $P_{a,a'}$ are contained in the union of the two paths in $T'$, from $a$ to the root of $T'$, and from $a'$ to the root of $T'$. Using Lemma \ref{lemma:ExternalEdges}, we get
\begin{eqnarray*}
w(P_{a,a_0})&=&w(P_{a,a'})+w(a',a_0)\\
&\leq&\frac{3^{i-1}-1}{2}(w(a,b)+w(a',a_0))+w(a',a_0)\\
&\leq&\frac{3^{i-1}-1}{2}(w(a,b)+w(a,b))+w(a,b)\\
&=&3^{i-1}w(a,b)~.
\end{eqnarray*}
Note that the inequality $w(P_{a,a_0})\leq3^{i-1}w(a,b)$ is true even if $P_{a,a_0}$ is empty. With an analogous proof, we can prove that $w(P_{b_0,b})\leq3^{i-1}w(a,b)$ as well. In addition, we can prove that if $P_{b_0,b}$ is not empty, and $b'$ is the next vertex after $b_0$ in $P$, then $w(b_0,b')\leq w(a,b)$. The proof is again analogous.

Generalizing the previous notations, let $a'$ be the preceding vertex to $a_0$ in $P$, where if $a_0=a$ (i.e., $P_{a,a_0}$ is empty), set $a'=b$. Let $b'$ be the next vertex after $b_0$ in $P$, where if $b_0=b$, set $b'=a$. Note that the edge $\{a_0,a'\}$ is always an edge of $G$ such that $a'$ is not in $T_0$; if $a_0\neq a$ then it is implied by the definition of $a_0$, and otherwise $a_0=a,a'=b$, so since $a$ is in $T_0$, $b$ cannot be in $T_0$. Symmetrically, $\{b_0,b'\}$ is an edge of $G$ such that $b'$ is not in $T_0$. Thus, by Lemma \ref{lemma:ExternalEdges},
\[
w(P_{a_0,b_0})\leq\frac{3^{i-1}-1}{2}(w(a_0,a')+w(b_0,b'))\leq(3^{i-1}-1)w(a,b)~.
\]

We finally conclude that
\begin{eqnarray*}
w(P)&=&w(P_{a,a_0})+w(P_{a_0,b_0})+w(P_{b_0,b})\\
&\leq&3^{i-1}w(a,b)+(3^{i-1}-1)w(a,b)+3^{i-1}w(a,b)\\
&=&(3^i-1)w(a,b)
\end{eqnarray*}

\end{proof}

The following lemma generalizes Lemma \ref{lemma:InternalEdges} to cases where there is no edge between the two vertices $a,b$.

\begin{lemma} \label{lemma:InternalPaths}
Let $a,b\in V$ be vertices of $G$ in the same tree $T$ of $F_i$, for some $i\in[0,l]$. Let $z=lca(a,b)$ be the lowest common ancestor in $T$, and let $P_{a,z},P_{b,z}$ be the unique paths in $T$ from $a$ to $z$ and from $b$ to $z$ respectively. Then,
\[
w(P_{a,z})+w(P_{b,z})\leq(3^i-1)d_G(a,b)~.
\]
\end{lemma}
\begin{proof}

As in the proof of Lemma \ref{lemma:InternalEdges}, it is enough to prove the lemma only for the \textit{minimal} index $i\in[0,l]$ such that $a,b$ are in the same tree of $F_i$. Thus, from now on we assume that $i$ is this minimal index.

We define a new graph $G_+$ which is identical to $G$, but with an additional edge $\{a,b\}$ between $a,b$ with weight $d_G(a,b)$. In case there was already an edge $\{a,b\}\in E$, we only change its weight to $d_G(a,b)$ if necessary.

We will show that performing the same construction of the forests $F_0,F_1,...,F_i$ on this new graph $G_+$, keeps the forests the same as before, and so the paths $P_{a,z},P_{b,z}$ in $T$ are the same paths. Then, we can use Lemma \ref{lemma:InternalEdges} to conclude our wanted result.

From now on, we assume that $G$ does not have an edge between $a,b$ with weight $d_G(a,b)$, as otherwise $G_+=G$ and the lemma follows directly from Lemma \ref{lemma:InternalEdges}. Thus, there is a shortest path $P_{a,b}$ in $G$ with at least $2$ edges between $a$ and $b$ with a total weight of $d_G(a,b)$. Since the weights of the edges of $G$ are positive, the weight of every edge in $P_{a,b}$ is {\em strictly less} than $d_G(a,b)$.

The proof that the forest $F_j$, when built on the graph $G_+$, is identical to $F_j$, when built on $G$, is by induction over $j\in[0,i]$. For $j=0$, the claim is trivial.

Let $j<i$, and assume that $F_j$ is identical in both cases. We prove that so is $F_{j+1}$. Let $T_a,T_b$ be the trees of $F_j$ that contain $a,b$ respectively. By the minimality of $i$, we know that $T_a\neq T_b$. The vertices of the graph $\mathcal{H}_j$ are still the trees of $F_j$, the same as before. The edges of $\mathcal{H}_j$ are all the edges of the form $\{T,T'\}$, where there is an edge in $E$ between a vertex of $T$ and a vertex of $T'$. The weight of $\{T,T'\}$ is the minimal weight of an edge in $E$ between $T,T'$. Hence, all the edges that were before in $\mathcal{H}_j$ are still there, with the same weights, since we did not delete edges from $E$. The only exception can occur in the edge $\{T_a,T_b\}$, where this edge could be either added to $\mathcal{H}_j$, or change its weight to $w(a,b)=d_G(a,b)$. Thus, in case $\mathcal{H}_j$ changes, the weight of this edge is $w(T_a,T_b)=w(a,b)=d_G(a,b)$.

As a result, for every $T\notin\{T_a,T_b\}$ in $\mathcal{H}_j$, the adjacent edges to $T$ are the same, and therefore $e_T$ is the same. We also claim that $e_{T_a},e_{T_b}$ are the same as before. To see that, note that the shortest path $P_{a,b}$ between $a,b$ must contain a vertex outside of $T_a$, since $b$ is not in $T_a$. Let $y$ be the first vertex on $P_{a,b}$ (in the direction from $a$ to $b$) that is not in $T_a$, and let $x$ be its preceding vertex. Denote by $T'$ the tree of $F_j$ that contains $y$. The edge $\{x,y\}\in E$ is between $T_a,T'$, thus $\{T_a,T'\}$ is in $\mathcal{H}_j$ and
\[
w(T_a,T')=w(e(T_a,T'))\leq w(x,y)<d_G(a,b)=w(a,b)~.
\]

Hence, if there was a change in the graph $\mathcal{H}_j$, and so $w(T_a,T_b)=w(a,b)$, it cannot be that $e_{T_a}$ is this new edge $\{T_a,T_b\}$. Therefore, $e_{T_a}$ equals to one of the original adjacent edges to $T_a$ in $\mathcal{H}_j$, and so it must be equal to the original $e_{T_a}$. With the same proof, we can show that $e_{T_b}$ stays the same as before, and thus all of the edges $e_T$, where $T$ is in $F_j$, stay the same.

By the constructions of the Bor\r{u}vka forest of $\mathcal{H}_j$, and the partial Bor\r{u}vka forest $\mathcal{H}''_j$, we can see that since the edges $e_T$ are the same, then also these graphs are the same. Recall that $F_{j+1}$ is defined based only on $\mathcal{H}''_j$ and $F_j$. We conclude that $F_{j+1}$ is identical to the forest that would have been constructed for the graph $G$, thus complete our proof.

The lemma now follows from Lemma \ref{lemma:InternalEdges}.

\end{proof}

In the next lemma, we use the notations $p_i(x)$ and $h_i(x)$ to denote the parent and the height of $x$ in the tree of $F_i$ that contains $x$. 


\begin{lemma*}[Lemma \ref{lemma:ExtractingAlgorithm} from Section \ref{sec:StretchAnalysisWeighted}] 
There is an algorithm that given two vertices $u,v\in V$, and a simple path $Q=(S_1,S_2,...,S_t)$ in the graph $\mathcal{H}_i$, such that $u$ is in $S_1$ and $v$ is in $S_t$, returns a path $P$ in $G$ between $u$ and $v$, with
\[
w(P)\leq3^{i+1}(d_G(u,v)+w(Q))~.
\]

The running time of the algorithm is proportional to the number of edges in the output path $P$. The required information for the algorithm is the set $\{h_i(x),p_i(x)\}_{x\in V}$, and the set $\{e(S_j,S_{j+1})\}_{j=1}^{t-1}$.
\end{lemma*}

\begin{proof}[Proof of Lemma \ref{lemma:ExtractingAlgorithm}]

First, notice that if $t=1$, i.e., the path $Q$ is an empty path from $S_1$ to itself, then by Lemma \ref{lemma:InternalPaths} the unique path in $S_1$ between $u,v$ has a weight of at most
\[
(3^i-1)d_G(u,v)<3^{i+1}(d_G(u,v)+w(Q))~.
\]
Thus, we can just set $P$ to be this unique path, and the claim will hold. To find $P$, use the algorithm from Lemma \ref{lemma:TreeRouting}, which requires $O(|P|)$ time.

In case that $t>1$, first set $\{x_j,y_j\}=e(S_j,S_{j+1})$ for every $j\in[1,t-1]$, where $x_j$ is in $S_j$ and $y_j$ is in $S_{j+1}$. For $j=0$ define $y_0=u$, and for $j=t$ define $x_t=v$. Notice that $x_j$ and $y_j$ are never in the same tree of $F_i$, since we assume that $S_j\neq S_{j+1}$ for every $j\in[1,t-1]$. Then, for every $j\in[1,t]$ the algorithm uses the variables $\{h_i(x),p_i(x)\}_{x\in V}$ and the algorithm from Lemma \ref{lemma:TreeRouting} to find the unique path $P_j$ in $S_j$ between $y_{j-1}$ and $x_j$. This takes $O(|P_j|)$ time. By Lemma \ref{lemma:ExternalEdges} we know that for every $j\in[2,t-1]$,
\[
w(P_j)\leq\frac{3^i-1}{2}(w(x_{j-1},y_{j-1})+w(x_j,y_j))~.
\]

The returned path $P$ is then
\[
P_1\circ\{x_1,y_1\}\circ P_2\circ\{x_2,y_2\}\circ\cdots\circ P_{t-1}\circ\{x_{t-1},y_{t-1}\}\circ P_t~.
\]
See figure \ref{fig:ClustersPath} for an illustration.

\begin{center}
\begin{figure}[ht!]
    \centering
    \includegraphics[width=12cm, height=4.5cm]{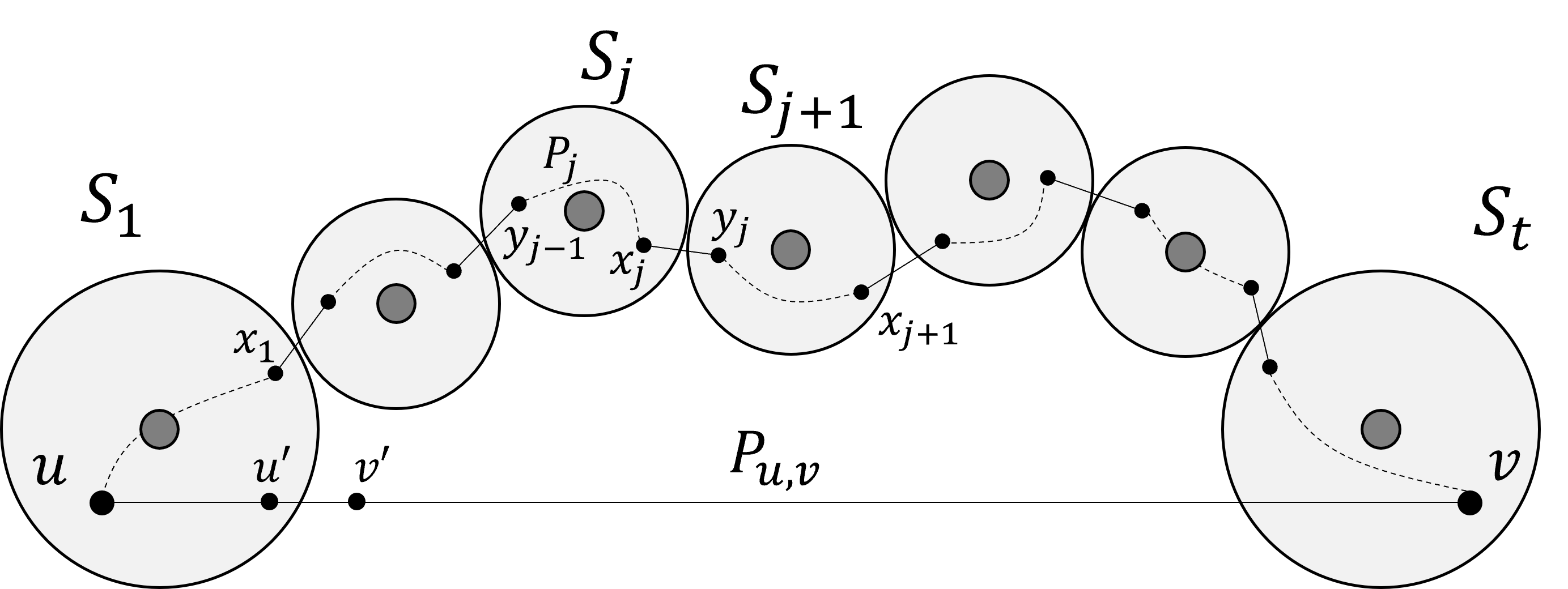}
    \caption{An illustration of Lemma \ref{lemma:ExtractingAlgorithm}. The dashed lines represent the paths $P_j$. The path denoted by $P_{u,v}$ is the shortest path between $u,v$ in the graph $G$. The vertex $v'$ is the first in this path to be outside of $S_1$, and $u'$ is its preceding vertex.}
    \label{fig:ClustersPath}
\end{figure}
\end{center}

Note that the running time of this algorithm is $\sum_{j=1}^tO(|P_j|)=O(|P|)$.

We now bound the weight of the path $P_1$. The bound for the weight of $P_t$ can be proved analogously. Let $P_{u,v}$ be the shortest path in $G$ between $u$ and $v$. Since $Q$ is simple, the vertex $v$ cannot be in $S_1$, otherwise $S_1=S_t$. Therefore $P_{u,v}$  pass through vertices that are not in $S_1$. Let $v'$ be the first vertex on $P_{u,v}$ (with the direction from $u$ to $v$) that is not in $S_1$, and let $u'$ be its preceding vertex on $P_{u,v}$. Note that $d_G(u,u')\leq d_G(u,v)$.

Denote by $P'_{u,u'}$ the unique path in $S_1$ between $u$ and $u'$, and by $P'_{u',x_1}$ the unique path in $S_1$ between $u'$ and $x_1$. The concatenation $P'_{u,u'}\circ P'_{u',x_1}$ is a path in $S_1$, not necessarily simple, between $u$ and $x_1$, and thus it contains $P_1$. Our bound on $w(P_1)$ will be achieved by bounding $w(P'_{u,u'})$ and $w(P'_{u',x_1})$.

For $P'_{u,u'}$, we use Lemma \ref{lemma:InternalPaths}:
\[
w(P'_{u,u'})\leq(3^i-1)d_G(u,u')\leq(3^i-1)d_G(u,v)~.
\]
For $P'_{u',x_1}$, we notice that $u',x_1$ are in the same tree $S_1$, and their adjacent edges $\{u',v'\}$ and $\{x_1,y_1\}$ satisfy $v',y_1\notin S_1$. Then, by Lemma \ref{lemma:ExternalEdges},
\[
w(P'_{u',x_1})\leq\frac{3^i-1}{2}(w(u',v')+w(x_1,y_1))\leq\frac{3^i-1}{2}(d_G(u,v)+w(x_1,y_1))~.
\]

We conclude that
\begin{eqnarray*}
w(P_1)&\leq&w(P'_{u,u'})+w(P'_{u',x_1})\\
&\leq&(3^i-1)d_G(u,v)+\frac{3^i-1}{2}(d_G(u,v)+w(x_1,y_1))\\
&=&\frac{3}{2}(3^i-1)d_G(u,v)+\frac{3^i-1}{2}w(x_1,y_1)~,
\end{eqnarray*}
and we similarly can prove that
\[
w(P_t)\leq\frac{3}{2}(3^i-1)d_G(u,v)+\frac{3^i-1}{2}w(x_{t-1},y_{t-1})~.
\]

Summarizing all of our results, we get
\begin{flalign*}
&w(P)&=\;&\sum_{j=1}^tw(P_j)+\sum_{j=1}^{t-1}w(x_j,y_j)&\\
&&\leq\;&\left[\frac{3}{2}(3^i-1)d_G(u,v)+\frac{3^i-1}{2}w(x_1,y_1)\right]&(P_1)\\
&&&+\;\sum_{j=2}^{t-1}\frac{3^i-1}{2}(w(x_{j-1},y_{j-1})+w(x_j,y_j))&(\{P_j\}_{j=2}^{t-1})\\
&&&+\;\left[\frac{3}{2}(3^i-1)d_G(u,v)+\frac{3^i-1}{2}w(x_{t-1},y_{t-1})\right]+\sum_{j=1}^{t-1}w(x_j,y_j)&(P_t)\\
&&=\;&3(3^i-1)d_G(u,v)+\sum_{j=1}^{t-1}3^iw(x_j,y_j)&\\
&&<\;&3^{i+1}(d_G(u,v)+\sum_{j=1}^{t-1}w(x_j,y_j))&\\
&&=\;&3^{i+1}(d_G(u,v)+w(Q))~.&\\
\end{flalign*}
The last step is due to the fact that $\{x_j,y_j\}=e(S_j,S_{j+1})$. This completes the proof of the lemma.

\end{proof}

\end{document}